\documentclass[a4paper,12pt]{article}
\usepackage{eurosym}
\usepackage[dvips]{graphicx}
\usepackage{amsfonts}
\usepackage{mathrsfs}
\usepackage{graphicx} 
\usepackage{epsfig}
\usepackage{amsmath, amsthm}
\usepackage{amssymb}

\newcommand{\ba}{\begin{array}}
\newcommand{\ea}{\end{array}}

\newcommand{\beq}{\begin{equation}}
\newcommand{\eeq}{\end{equation}}
\newcommand{\ben}{\begin{enumerate}}
\newcommand{\een}{\end{enumerate}}
\newcommand{\bit}{\begin{itemize}}
\newcommand{\eit}{\end{itemize}}

\newtheorem{lem}{Lemma}
\newtheorem{prop}{Proposition}
\newtheorem{theorem}{Theorem}
\theoremstyle{remark}
\newtheorem*{remark}{Remark}
\begin{document}
\title{General bright and dark soliton solutions to the massive Thirring model via KP hierarchy reductions}
\author{Junchao Chen \\
\textsl{Department of Mathematics, Lishui University, China}
\and Bao-Feng Feng \\
\textsl{School of Mathematical and Statistical Sciences} \\
\textsl{The University of Texas Rio Grande Valley}}
\maketitle

\begin{abstract}
In the present paper, we are concerned with the tau function and its connection with the Kadomtsev-Petviashvili (KP) theory for the massive Thirring (MT) model. First, we bilinearize the massive Thirring model under both the vanishing and nonvanishing boundary conditions. Starting from a set of bilinear equations of two-component KP-Toda hierarchy, we derive the multi-bright solution to the MT model by the KP hierarchy reductions. Then, we show that the discrete KP equation can generate a set of bilinear equations of a deformed KP-Toda hierarchy through Miwa transformation. By imposing constraints  on the parameters of the tau function, the general dark soliton solution to the MT model is  constructed from the tau function of the discrete KP equation. Finally, the dynamics and properties of one- and two-soliton for both the bright and dark cases are analyzed in details.
\end{abstract}
\section{Introduction}
The massive Thirring (MT) model
\begin{eqnarray}
&&  \displaystyle \mathrm{i} u_x + v +\sigma u|v|^2 =0\,,  \label{MTa} \\ [5pt]
&& \displaystyle \mathrm{i} v_t + u + \sigma v|u|^2  =0 \label{MTb} \,,
\end{eqnarray}
with $\sigma=\pm1$, was derived by Thirring in 1958  \cite{Thirring} in the context of
general relativity. It represents a relativistically invariant nonlinear Dirac equation in the space of
one dimension. It is one of the most remarkable solvable field theory models. Its complete integrability was firstly approved by
Mikhailov \cite{Mikhailov} and Orfanidis \cite{Orfanidis} independently.  The inverse scattering transform for the MT model was studied by
Kuznetsov and Mikhailov \cite{Kuznetsov} and many others \cite{Kawata79,WadatiMT83,KaupLakoba,Villarroel:1991,DmkitryMTIST}. The Darboux transformation, B\"acklund transformation of the MT model and its connection with other integrable systems have been investigated by Kaup and Newell \cite{KaupNewell}, Lee \cite{Lee:1993,Lee:1994}, Prikarpatskii \cite{Prikarpatskii:1979,Prikarpatskii:1981} , Franca {\it et al.} \cite{Franca13} and Degasperis \cite{DegasperisMT}.

Since the pioneer work by Date \cite{Date}, the soliton solutions to the MT model including the ones with nonvanishing background were constructed by many authors
\cite{Shnider84,Alonso:1984,BarashenkovGetmanov:1987,BarashenkovGetmanovKovtun:1993,BarashenkovGetmanov:1993,Talalov:1987,Vaklev:1996}. In addition,
the algebro-geometric solutions to the MT model has attracted attention and has been studied from late 1970s \cite{DateMT2,MTquasi2} to 1980s \cite{Bikbaev,Wisse85} until recently \cite{MTEnolskii,MTElbeck}.
It should be pointed out that the rogue wave solutions to the MT model were recently investigated in \cite{DWA-MT2015,HeMT2015,ShihuaMT} by the Darboux transformation.

The MT model admits the following Lax pair
\begin{equation}
\Phi_x= U\Phi, \ \quad
\Phi_t = V \Phi, \label{matrix-MT}
\end{equation}
where
\begin{equation*}
 U=-\frac{\mathrm{i}}{2} \left(
 \begin{matrix}
\lambda^2 - \sigma vv^* &  2\lambda v \\
 2\lambda \sigma v^* & -\lambda^2 + \sigma vv^*
 \end{matrix}
 \right)
\end{equation*}
\begin{equation*}
V=-\frac{\mathrm{i}}{2} \left(
 \begin{matrix}
\lambda^{-2} -\sigma uu^* &  2\lambda^{-1} u \\
 2\lambda^{-1}\sigma u^* & -\lambda^{-2} + \sigma uu^*
 \end{matrix}
 \right)\,,
\end{equation*}
in the sense that the compatibility condition $U_t-V_x + [U,V]=0$ gives the MT model (\ref{MTa})--(\ref{MTb}).
It is noted that the MT model can also be written in the laboratory coordinates  as follows
\begin{eqnarray}
&&  \displaystyle \mathrm{i} (u_T + u_X)+ v + \sigma u|v|^2 =0\,,  \label{MTLa} \\ [5pt]
&& \displaystyle \mathrm{i} (v_T - v_X) + u + \sigma v|u|^2  =0 \label{MTLb} \,,
\end{eqnarray}
via the following transformation
 \begin{equation*}
X=x+t\,, \quad T=x-t\,.
\end{equation*}
The MT model in the laboratory coordinates has application in nonlinear optics to describe pulse propagation in Bragg gratings \cite{Sipe}.  Alexeeva {\it et al.} studied the PT-symmetry extensions of the MT model \cite{BarashenkovGetmanov:2019}.
Regarding the integrable discretization, Nijhoff {\it et al.} gave the integrable discretization of the MT model in light cone coordinates \cite{NijhoffMT1,NijhoffMT2} in 1980s. Most recently, Pelinovsky {\it et al.} proposed a semi-discrete integrable MT model in laboratory coordinates \cite{DmkitryMTdiscrete} and studied its solution solution via Darboux transformation \cite{DmitryXu2019}.

Surprisingly, as far as we are aware, the bilinear formulation is missing in the literature. As far as we are aware, no paper published either in finding solutions via Hirota's bilinear method \cite{Hirotabook} or revealing the connection of the MT model to the KP theory by Kyoto school \cite{JM}.  By combining the Hirota's bilinear method and the KP hierarchy reduction method, we have constructed general soliton solutions to many soliton equations such as the vector nonlinear Schr\"odiner equation \cite{FengvNLS} and the complex short pulse equation \cite{FengShen_ComplexSPE,FMO_ComplexSPE} for both the vanishing boundary condition (VBC) and nonvanishing boundary condition (NVBC).
 Therefore, the motivation and the goal of the present study is to bilinearize the MT model and find its soliton solutions under VBC and NVBC. 
 
 The remainder of the present paper in organized as follow. In section 2, we first bilinearize the MT model into a set of four  equations corresponding to the bright soliton solution. Starting from a set of bilinear equations satisfied by the tau functions of two-component KP hierarchy, we arrive at above four bilinear equations satisfied by the bright soliton solutions by a series of reductions such as dimension and complex conjugate reductions.  On the other hand, in section 3, we show that the discrete KP equation can generate a single-component KP-Toda hierarchy with four flows: one in positive flow and three in negative flows, along with its tau-function and a set of bilinear equations. By imposing one constraint on parameters, three reduction relations are achieved simultaneously, by which we are finally able to construct general $N$-dark soliton solution to the MT model. The dynamical behaviors and properties of one- and two-soliton solutions for both the bright and dark ones are analyed in section 4. The paper is concluded in section 5 by some comments and further topics.

\section{Bright solitons in the MT model}
\subsection{Bilinearization of the MT model under VBC}
The bilinearization of the MT system is established
by the following proposition.
\begin{prop}
By means of the dependent variable transformations
\begin{equation} \label{bt_tran1}
u=\frac{g}{f^{\ast}}\,, \quad
v=\frac{h}{f}\,,
\end{equation}
the MT model (\ref{MTa})--(\ref{MTb}) is transformed into the following bilinear equations
\begin{eqnarray}
&& \displaystyle \mathrm{i} D_{x}g \cdot f +  h f ^{\ast}  =0   \,,  \label{MTbtBL1} \\ [5pt]
&& \displaystyle \mathrm{i}D_{x}  f  \cdot f^{\ast}  =-\sigma hh^{\ast} \,, \label{MTbtBL2}  \\ [5pt]
&& \displaystyle \mathrm{i} D_{t} h \cdot f^{\ast}   +  g f=0 \,,   \label{MTbtBL3} \\ [5pt]
&&\displaystyle  \mathrm{i} D_{t} f \cdot f^{\ast } = \sigma gg^{\ast}  \,. \label{MTbtBL4}
\end{eqnarray}
where $D$ is the Hirota $D$-operator defined by \cite{Hirotabook}
\begin{equation*}
D_s^n D_y^m f\cdot g=\left(\frac{\partial}{\partial s} -\frac{\partial}{%
\partial s^{\prime }}\right)^n \left(\frac{\partial}{\partial y} -\frac{%
\partial}{\partial y^{\prime }}\right)^m f(y,s)g(y^{\prime },s^{\prime
})|_{y=y^{\prime }, s=s^{\prime }}\,.
\end{equation*}
\end{prop}
\begin{proof}
By rewriting the dependent variable transformations (\ref{bt_tran1})
\begin{equation*}
u=\frac{g}{f}   \frac{f}{f^{\ast} }  \,, \quad
v=\frac{{h}}{f^{\ast}} \frac{f^{\ast}}{{f}}\,,
\end{equation*}
and substituting into Eq.(\ref{MTa}), we have
 \begin{eqnarray}
&&\left[
  \mathrm{i}  \left(\frac{g}{f}\right)_x \frac{f}{f^{\ast} }  +  \frac{h}{f}
\right]
+\left[
  \mathrm{i} \frac{g}{f} \left( \frac{f}{f^{\ast} }\right)_x
 +   \sigma \frac{hh^*}{ff^*}  \frac{g}{f^{\ast} } \right]   =0 \,.
\end{eqnarray}
Bilinear equations (\ref{MTbtBL1}) and  (\ref{MTbtBL2}) are deduced by taking zero for each group inside bracket. Similarly, we can drive bilinear equations
(\ref{MTbtBL3}) and  (\ref{MTbtBL4}) by substituting (\ref{bt_tran1})  into Eq.(\ref{MTb}).
\end{proof}

\subsection{Reductions to the multi-bright solution from the two-component KP-Toda hierarchy}
Define the following tau functions of two-component KP-Toda
hierarchy,
\begin{equation*}
f_{nm}=\left\vert
\begin{array}{cc}
A_n & I \\
-I & B_m\end{array}\right\vert \,,
\end{equation*}
\begin{equation*}
g_{nm}=
\left\vert
\begin{array}{ccc}
A_n & I & \Phi_n ^{T} \\
-I & B_m & \mathbf{0}^{T} \\
\mathbf{0} & -\bar{\Psi}_m & 0\end{array}\right\vert \,,\quad
\bar{g}_{nm}=\left\vert
\begin{array}{ccc}
A'_n & I & \mathbf{0}^{T} \\
-I & B_m & {\Psi }_m^{T} \\
-\bar{\Phi}_n & \mathbf{0} & 0\end{array}\right\vert \,,
\end{equation*}
where $A_n$ and $B_m$ are $N\times N$ matrices whose elements are
\begin{equation*}
a_{ij}(n)=\frac{\mu \bar{p}_{j} }{p_{i}+\bar{p}_{j}}\left( -\frac{p_{i}}{\bar{p}_{j}}\right)
^{n}e^{\xi _{i}+\bar{\xi}_{j}},
\quad
a'_{ij}(n)=-\frac{\mu p_{i} }{p_{i}+\bar{p}_{j}}\left( -\frac{p_{i}}{\bar{p}_{j}}\right)
^{n}e^{\xi _{i}+\bar{\xi}_{j}},
\end{equation*}
\begin{equation*}
b_{ij}(m)=\frac{\nu}{q_{i}+\bar{q}_{j}}\left( -\frac{q_{i}}{\bar{q}_{j}}\right) ^{m}e^{\eta _{i}+\bar{\eta}_{j}}\,,
\end{equation*}
with
\begin{equation*}
\xi _{i}=\frac{1}{p_{i}}x_{-1}+p_{i}x_{1}+\xi _{i0},\quad \bar{\xi}_{j}=\frac{1}{\bar{p}_{j}}x_{-1}+\bar{p}_{j}x_{1}+\bar{\xi}_{j0},
\end{equation*}
\begin{equation*}
\eta _{i}=\frac{1}{q_{i}}y_{-1}+q_{i}y_{1}+\eta _{i0},\quad \bar{\eta}_{j}=\frac{1}{\bar{q}_{j}}y_{-1}+\bar{q}_{j}y_{1}+\bar{\eta}_{j0},
\end{equation*}
where $\Phi_n $, $\Psi_n$, $\bar{\Phi}_n$ and $\bar{\Psi}_n$ are $N$-component row vectors
\begin{equation*}
\Phi_n =\left(p_1^ne^{\xi _{1}},\cdots ,p_N^ne^{\xi _{N}}\right) \,,\ \bar{\Phi}_n=\left((-\bar{p}_{1})^{-n}e^{\bar{\xi}_{1}},\cdots ,(-\bar{p}_{N})^{-n}e^{\bar{\xi}_{N}}\right) \,,\
\end{equation*}\begin{equation*}
\Psi_m =\left(q_1^me^{\eta _{1}},\cdots ,q_N^me^{\eta _{N}}\right) \,, \bar{\Psi}_m=\left((-\bar{q}_{1})^{-m}e^{\bar{\eta}_{1}},\cdots, (-\bar{q}_{N})^{-m}e^{\bar{\eta}_{N}}\right) \,.
\end{equation*}
Then we have the following lemma:
\begin{lem}
The above tau functions of two-component KP-Toda hierarchy satisfy the following bilinear equations
\begin{eqnarray}
\label{before-kp-1} &&D_{x_{-1}}g_{nm}\cdot f_{nm}=  g_{n-1,m}f_{n+1,m},
\\
\label{before-kp-2} &&D_{x_1}g_{n,m+1}\cdot f_{n+1,m}=  g_{n+1,m+1}f_{nm},
\\
\label{before-kp-3} &&D_{y_1}f_{n+1,m}\cdot f_{n,m}=  \mu \nu g_{n,m}\bar{g}_{nm},
\\
\label{before-kp-4} &&D_{y_{-1}}f_{n+1,m}\cdot f_{n,m}=  -\mu \nu g_{n,m+1}\bar{g}_{n,m-1}.
\end{eqnarray}
\end{lem}
\begin{proof}
Let us take the following notations:
\begin{eqnarray*}
&&
\phi_i(n)=p^n_ie^{\xi_i},\ \
\bar{\phi}_i(n)=(-\bar{p}_{i})^{-n}e^{\bar{\xi}_{i}},\ \
\\
&&
\psi_i(m)=q^m_ie^{\eta_i},\ \
\bar{\psi}_i(m)=(-\bar{q}_{i})^{-m}e^{\bar{\eta}_{i}},\ \
\end{eqnarray*}
then the above matrix elements possess the differential and difference rules:
\begin{eqnarray*}
&&
\partial_{x_1} a_{ij}(n)=-\mu \phi_i(n)\bar{\phi}_j(n-1),\ \
\partial_{x_{-1}} a_{ij}(n)=\mu \phi_i(n-1)\bar{\phi}_j(n),
\\
&&
a_{ij}(n+1)=a_{ij}(n)-\mu\phi_i(n)\bar{\phi}_j(n),
\\
&&
\partial_{y_1} b_{ij}(m)=\nu \psi_i(m)\bar{\psi}_j(m),\ \
\partial_{y_{-1}} b_{ij}(m)=-\nu \psi_i(m-1)\bar{\psi}_j(m+1),
\\
&&
b_{ij}(m+1)=b_{ij}(m)+\nu\psi_i(m)\bar{\psi}_j(m+1),
\\
&&
\partial_{x_s} \phi_i(n)=\phi_i(n+s), \ \
\partial_{x_s} \bar{\phi}_i(n)=-\bar{\phi}_i(n-s),\ \
\\
&&
\partial_{y_s} \psi_i(m)=\psi_i(m+s), \ \
\partial_{y_s} \bar{\psi}_i(m)=-\bar{\psi}_i(m-s),\ \ (s=\pm 1).
\end{eqnarray*}
By using the differential formula of determinant, it can be
checked that the derivatives and shifts of the tau functions are expressed by the bordered determinants as follows:
\begin{eqnarray*}
&&
f_{n+1,m}=\left\vert
\begin{array}{ccc}
A_n & I & \Phi_n ^{T} \\
-I & B_m & \mathbf{0}^{T} \\
\mu\bar{\Phi}_n & \mathbf{0} & 1
\end{array}\right\vert \,, \ \
\partial_{x_{-1}}f_{n,m}=\left\vert
\begin{array}{ccc}
A_n & I & \Phi_{n-1} ^{T} \\
-I & B_m & \mathbf{0}^{T} \\
-\mu\bar{\Phi}_n & \mathbf{0} & 0
\end{array}\right\vert \,,\ \
\\
&&
\partial_{x_{1}}f_{n+1,m}
=\left\vert
\begin{array}{ccc}
A_n & I & \Phi_{n+1} ^{T} \\
-I & B_m & \mathbf{0}^{T} \\
\mu\bar{\Phi}_n & \mathbf{0} & 0
\end{array}\right\vert \,,\ \
\partial_{y_{1}}f_{n,m}
=\left\vert
\begin{array}{ccc}
A_n & I & \mathbf{0}^{T} \\
-I & B_m & \Psi_{m}^{T} \\
\mathbf{0} & -\nu\bar{\Psi}_m  & 0
\end{array}\right\vert \,,
\\
&&
\partial_{y_{-1}}f_{n,m}
=\left\vert
\begin{array}{ccc}
A_n & I & \mathbf{0}^{T} \\
-I & B_m & \Psi_{m-1}^{T} \\
\mathbf{0} & \nu\bar{\Psi}_{m+1}  & 0
\end{array}\right\vert \,,
\ \
\partial_{y_{1}}f_{n+1,m}
=\left\vert
\begin{array}{cccc}
A_n & I & \Phi_{n}^{T} &\mathbf{0}^{T} \\
-I & B_m & \mathbf{0}^{T} & \Psi_{m}^{T} \\
\mu \bar{\Phi}_{n} & \mathbf{0}^{T} & 1 & 0\\
\mathbf{0} & -\nu\bar{\Psi}_m  & 0 & 0
\end{array}\right\vert \,,
\\
&&\partial_{y_{-1}}f_{n+1,m}
=\left\vert
\begin{array}{cccc}
A_n & I & \Phi_{n}^{T} &\mathbf{0}^{T} \\
-I & B_m & \mathbf{0}^{T} & \Psi_{m-1}^{T} \\
\mu \bar{\Phi}_{n} & \mathbf{0}^{T} & 1 & 0\\
\mathbf{0} & \nu\bar{\Psi}_{m+1}  & 0 & 0
\end{array}\right\vert \,,
\ \
g_{n-1,m}
=\left\vert
\begin{array}{ccc}
A_n & I & \Phi_{n-1}^{T}  \\
-I & B_m & \mathbf{0}^{T} \\
\mathbf{0} & -\bar{\Psi}_{m}  & 0
\end{array}\right\vert \,,
\\
&&g_{n,m+1}
=\left\vert
\begin{array}{ccc}
A_n & I & \Phi_{n}^{T}  \\
-I & B_m & \mathbf{0}^{T} \\
\mathbf{0} & -\bar{\Psi}_{m+1}  & 0
\end{array}\right\vert \,,
\ \
g_{n+1,m+1}
=\left\vert
\begin{array}{cccc}
A_n & I & \Phi_{n+1}^{T} & \Phi_{n}^{T}  \\
-I & B_m & \mathbf{0}^{T} & \mathbf{0}^{T} \\
\mathbf{0} & -\bar{\Psi}_{m+1}  & 0 & 0\\
\bar{\Phi}_{n} &\mathbf{0} & 0 & 1
\end{array}\right\vert \,,
\end{eqnarray*}
\begin{eqnarray*}
&&
\partial_{x_{-1}} g_{n,m}
=-\left\vert
\begin{array}{cccc}
A_n & I & \Phi_{n}^{T} & \Phi_{n-1}^{T}  \\
-I & B_m & \mathbf{0}^{T} & \mathbf{0}^{T} \\
\mathbf{0} & -\bar{\Psi}_{m}  & 0 & 0\\
\mu\bar{\Phi}_{n} &\mathbf{0} & 1 & 0
\end{array}\right\vert \,,
\ \
\partial_{x_1}g_{n,m+1}
=\left\vert
\begin{array}{ccc}
A_n & I & \Phi_{n+1}^{T}  \\
-I & B_m & \mathbf{0}^{T} \\
\mathbf{0} & -\bar{\Psi}_{m+1}  & 0
\end{array}\right\vert \,, 
\\
&&
\bar{g}_{n,m}
=\left\vert
\begin{array}{ccc}
A_{n+1} & I & \mathbf{0}^{T}   \\
-I & B_m & \Psi_{m}^{T} \\
-\bar{\Phi}_{n} & \mathbf{0} & 0
\end{array}\right\vert 
=\left\vert
\begin{array}{ccc}
A_{n} & I & \mathbf{0}^{T}   \\
-I & B_m & \Psi_{m}^{T} \\
-\bar{\Phi}_{n} & \mathbf{0} & 0
\end{array}\right\vert,
\\
&&
\bar{g}_{n,m-1}
=\left\vert
\begin{array}{ccc}
A_{n+1} & I & \mathbf{0}^{T}   \\
-I & B_{m-1} & \Psi_{m-1}^{T} \\
-\bar{\Phi}_{n} & \mathbf{0} & 0
\end{array}\right\vert
=\left\vert
\begin{array}{ccc}
A_{n} & I & \mathbf{0}^{T}   \\
-I & B_m & \Psi_{m-1}^{T} \\
-\bar{\Phi}_{n} & \mathbf{0} & 0
\end{array}\right\vert.
\end{eqnarray*}
By using the Jacobi identity of determinant, we can obtain the following relations
\begin{eqnarray*}
&&f_{nm} \partial_{x_{-1}} g_{nm} = g_{nm} \partial_{x_{-1}}f_{nm}+  g_{n-1,m}f_{n+1,m},
\\
&&f_{nm} g_{n+1,m+1}=f_{n+1,m} \partial_{x_1}g_{n,m+1}- g_{n,m+1}\partial_{x_1}f_{n+1,m},
\\
&&f_{n,m} \partial_{y_1}f_{n+1,m} = f_{n+1,m} \partial_{y_1}f_{n,m}+ \mu \nu g_{n,m}\bar{g}_{nm},
\\
&&f_{n,m} \partial_{y_{-1}}f_{n+1,m} =f_{n+1,m} \partial_{y_{-1}}f_{n,m}  -\mu \nu g_{n,m+1}\bar{g}_{n,m-1},
\end{eqnarray*}
which completes the proof of bilinear equations (\ref{before-kp-1})-(\ref{before-kp-4}).
\end{proof}
In what follows, we proceed to the reduction procedure. To this end,  we impose the reduction condition
\begin{equation}\label{reduction-conditon}
q_i=\bar{p}_i,\ \ \bar{q}_i=p_i,
\end{equation}
then one can easily show
\begin{equation}
f_{n+1,m+1}=f_{n,m}\,, \quad g_{n+1,m+1}=-g_{n,m}\,,
\end{equation}
which leads to
\begin{eqnarray}
\label{before-kp-11} && D_{x_{-1}}g_{nm}\cdot f_{nm}=  -g_{n,m+1}f_{n+1,m},
\\
\label{before-kp-21} && D_{x_1}g_{n,m+1}\cdot f_{n+1,m}=  -g_{n,m}f_{nm}
\end{eqnarray}
from Eqs. (\ref{before-kp-1})--Eqs. (\ref{before-kp-2}).
On the other hand, by performing row and column operations, we rewrite the tau function $f_{nm}$ as
\begin{equation*}
f_{nm}=\left\vert
\begin{array}{cc}
\tilde{A}_n & I \\
-I & \tilde{B}_m\end{array}\right\vert \,,
\end{equation*}
with the elements
\begin{equation*}
\tilde{a}_{ij}(n)=\frac{\mu \bar{p}_{j} }{p_{i}+\bar{p}_{j}}\left( -\frac{p_{i}}{\bar{p}_{j}}\right)
^{n},
\quad
\tilde{b}_{ij}(m)=\frac{\nu}{q_{i}+\bar{q}_{j}}\left( -\frac{q_{i}}{\bar{q}_{j}}\right) ^{m}e^{\eta _{i}+\bar{\eta}_{j}+\xi _{j}+\bar{\xi}_{i}}\,,
\end{equation*}
From the condition (\ref{reduction-conditon}), one can deduce
\begin{equation}
\partial_{x_1}\tilde{b}_{ij}(m)=\partial_{y_1}\tilde{b}_{ij}(m),\ \ \partial_{x_{-1}}\tilde{b}_{ij}(m)=\partial_{y_{-1}}\tilde{b}_{ij}(m),
\end{equation}
which implies
\begin{equation}
\partial_{x_1}f_{nm}=\partial_{y_1}f_{nm}\,, \quad \partial_{x_{-1}}f_{nm}=\partial_{y_{-1}}f_{nm}\,.
\end{equation}
Thus Eqs. (\ref{before-kp-3})--(\ref{before-kp-4}) become
\begin{eqnarray}
\label{before-kp-31} &&D_{x_1}f_{n+1,m}\cdot f_{n,m}=  \mu \nu g_{n,m}\bar{g}_{nm},
\\
\label{before-kp-41} &&D_{x_{-1}}f_{n+1,m}\cdot f_{n,m}=  -\mu \nu g_{n,m+1}\bar{g}_{n,m-1}.
\end{eqnarray}

Next, we consider complex conjugate condition by giving the following lemma:
\begin{lem}
Assume $\mu=-{\rm i}$, $\nu=\sigma$, $n=m=0$,
\begin{equation}
\bar{p}_i=p^*_i,\ \ \bar{\xi}_{i0}=\xi^*_{i0},\ \ e^{\eta_{i0}}=\alpha_i,\ \ e^{\bar{\eta}_{i0}}=\alpha^*_i,
\end{equation}
and drop the dummy variables $y_1$ and $y_{-1}$,  then
 \begin{equation}
a'_{ij}(0)=a^*_{ji}(0),\ \ b^*_{ij}(0)=b_{ji}(0),
 \end{equation}
which implies
\begin{equation}
f_{00}=f^*_{10},\ \ \bar{g}_{00}=-g^*_{00}, \ \ \bar{g}_{01}=g^*_{0,-1}.
\end{equation}
\end{lem}

\begin{proof}
Let $C=(\alpha_1,\alpha_2,\cdots,\alpha_N)$ be a row vector. Direct calculations give
\begin{eqnarray*}
f_{00}
&=&\left\vert
\begin{array}{cc}
-\frac{{\rm i} p^*_{j} }{p_{i}+p^*_{j}}e^{\xi _{i}+\xi^*_{j}} & I \\
-I & \frac{\sigma\alpha_i\alpha^*_j}{p^*_{i}+p_{j}}\end{array}\right\vert
=\left\vert
\begin{array}{cc}
-\frac{{\rm i} p^*_{i} }{p_{j}+p^*_{i}}e^{\xi _{j}+\xi^*_{i}} & -I \\
I & \frac{\sigma\alpha_j\alpha^*_i}{p^*_{j}+p_{i}}\end{array}\right\vert \\
&=&\left\vert
\begin{array}{cc}
-\frac{{\rm i} p^*_{i} }{p_{j}+p^*_{i}}e^{\xi _{j}+\xi^*_{i}} & I \\
-I & \frac{\sigma\alpha_j\alpha^*_i}{p^*_{j}+p_{i}}\end{array}\right\vert
=\left\vert
\begin{array}{cc}
\frac{{\rm i} p_{i} }{p_{i}+p^*_{j}}e^{\xi _{i}+\xi^*_{j}} & I \\
-I & \frac{\sigma\alpha_i\alpha^*_j}{p^*_{i}+p_{j}}\end{array}\right\vert^*
=f^*_{10},
\\
\bar{g}_{00}
&=&
\left\vert
\begin{array}{ccc}
\frac{{\rm i} p_{i} }{p_{i}+p^*_{j}}e^{\xi _{i}+\xi^*_{j}} & I & \mathbf{0}^{T} \\
-I & \frac{\sigma\alpha_i\alpha^*_j}{p^*_{i}+p_{j}} & C^{T} \\
-\Phi^* & \mathbf{0} & 0\end{array}\right\vert =
\left\vert
\begin{array}{ccc}
\frac{{\rm i} p_{j} }{p_{j}+p^*_{i}}e^{\xi _{j}+\xi^*_{i}} & -I & -\Phi^{*T} \\
I & \frac{\sigma\alpha_j\alpha^*_i}{p^*_{j}+p_{i}} & \mathbf{0}^{T} \\
\mathbf{0} & C & 0\end{array}\right\vert \\
&=&\left\vert
\begin{array}{ccc}
\frac{{\rm i} p_{j} }{p_{j}+p^*_{i}}e^{\xi _{j}+\xi^*_{i}} & I & -\Phi^{*T} \\
-I & \frac{\sigma\alpha_j\alpha^*_i}{p^*_{j}+p_{i}} & \mathbf{0}^{T} \\
\mathbf{0} & -C & 0\end{array}\right\vert
=-g^*_{00},
\\
\bar{g}_{0,-1}
&=&\left\vert
\begin{array}{ccc}
\frac{{\rm i} p_{i} }{p_{i}+p^*_{j}}e^{\xi _{i}+\xi^*_{j}} & I & \mathbf{0}^{T} \\
-I & \frac{\sigma\alpha_i\alpha^*_j}{p^*_{i}+p_{j}}\left(-\frac{p_j}{p^*_i}\right) & \frac{C^{T}}{p^*_i} \\
-\Phi^* & \mathbf{0} & 0\end{array}\right\vert
=\left\vert
\begin{array}{ccc}
\frac{{\rm i} p_{j} }{p_{j}+p^*_{i}}e^{\xi _{j}+\xi^*_{i}} & -I & -\Phi^{*T} \\
I & \frac{\sigma\alpha_j\alpha^*_i}{p^*_{j}+p_{i}}\left(-\frac{p_i}{p^*_j}\right) & \mathbf{0}^{T} \\
\mathbf{0} & \frac{C}{p^*_i} & 0
\end{array}\right\vert \\
&=&\left\vert
\begin{array}{ccc}
\frac{{\rm i} p_{j} }{p_{j}+p^*_{i}}e^{\xi _{j}+\xi^*_{i}} & I & -\Phi^{*T} \\
-I & \frac{\sigma\alpha_j\alpha^*_i}{p^*_{j}+p_{i}}\left(-\frac{p_i}{p^*_j}\right) & \mathbf{0}^{T} \\
\mathbf{0} & -\frac{C}{p^*_i} & 0\end{array}\right\vert
=g^*_{0,1}\,.
\end{eqnarray*}
\end{proof}
Finally, let $x_1=x$ and $x_{-1}=-t$, we have
\begin{eqnarray}
\label{before-kp-12} && D_{t}g_{00}\cdot f_{00}=  g_{01}f_{10},
\\
\label{before-kp-22} && D_{x}g_{01}\cdot f_{10}=  -g_{00}f_{00},
\\
\label{before-kp-32} &&D_{x}f_{10}\cdot f_{00}=  -{\rm i}\sigma g_{00}\bar{g}_{00},
\\
\label{before-kp-42} &&D_{t}f_{10}\cdot f_{00}=  -{\rm i}\sigma g_{01}\bar{g}_{0,-1}.
\end{eqnarray}
Therefore, we arrive at exactly the same set of bilinear equations  (\ref{MTbtBL1})--(\ref{MTbtBL4}) by setting
\begin{equation}
f_{00}=f^*,\ \ f_{10}=f,\ \ g_{00}=h,\ \ \bar{g}_{00}=-h^*,\ \ g_{01}={\rm i}g,\ \ \bar{g}_{0,-1}=-{\rm i}g^*.
\end{equation}
In summary, we can give the $N$-bright soliton solution to the MT model by the following theorem:
\begin{theorem}
The MT model  (\ref{MTa})--(\ref{MTb}) admits the multi-bright soliton solution
$u=\frac{g}{f^*}$, $v=\frac{h}{f}$ where $f$, $f^*$, $g$ and $h$ are the following determinant solution
\begin{eqnarray}
\label{multi-bright-sol-theorem-01}&&
f=\left\vert
\begin{array}{cc}
A & I \\
-I & B\end{array}\right\vert
,\ \
f^*
=\left\vert
\begin{array}{cc}
A' & I \\
-I & B\end{array}\right\vert,
\\
\label{multi-bright-sol-theorem-02}
&&
g=-{\rm i}
\left\vert
\begin{array}{ccc}
A' & I & \Phi^{T} \\
-I & B & \mathbf{0}^{T} \\
\mathbf{0} & \frac{C^*}{p_j} & 0\end{array}\right\vert,
\ \
h=
\left\vert
\begin{array}{ccc}
A' & I & \Phi^{T} \\
-I & B & \mathbf{0}^{T} \\
\mathbf{0} & -C^* & 0\end{array}\right\vert\,,
\end{eqnarray}
where $I$ is an $N\times N$ identity matrix, $A$, $A'$, and $B$ are $N\times N$ matrices
whose entries are
\begin{equation*}
a_{ij}=\frac{{\rm i} p_{i} }{p_{i}+p^*_{j}}e^{\xi _{i}+\xi^*_{j}},\ \
a'_{ij}=-\frac{{\rm i} p^*_{j} }{p_{i}+p^*_{j}}e^{\xi _{i}+\xi^*_{j}},\ \
b_{ij}=\frac{\sigma \alpha_i\alpha^*_j}{p^*_{i}+p_{j}},
\end{equation*}
and $\textbf{0}$ is a $N$-component zero-row vector, $\Phi$ and $C$ are N-component row
vectors given by
\begin{equation*}
\Phi=(e^{\xi_1},e^{\xi_2},\cdots,e^{\xi_N}),\ \ C=(\alpha_1,\alpha_2,\cdots,\alpha_N),
\end{equation*}
with $\xi _{i}=p_{i}x-\frac{1}{p_{i}}t+\xi _{i0}$.
Here $p_i$, $\alpha_i$ and $\xi _{i0}$ are arbitrary complex parameters for $i=1,\cdots,N$.
\end{theorem}

\section{Dark solitons in the MT model}
\subsection{Bilinearization of the MT model under NVBC}
The bilinearization of the MT model  (\ref{MTa})--(\ref{MTb}) under NVBC is established
by the following proposition.
\begin{prop}
By means of the dependent variable transformations
\begin{equation} \label{var_tran1}
u=\rho_1 \frac{g}{f^{\ast} } e^{\mathrm{i}(1+\sigma \rho_1\rho_2) \left(\frac{\rho_2}{\rho_1}x+\frac{\rho_1}{\rho_2}t\right)}\,, \quad
v=\rho_2 \frac{{h}}{{f}} e^{\mathrm{i}(1+\sigma \rho_1\rho_2) \left(\frac{\rho_2}{\rho_1}x+\frac{\rho_1}{\rho_2}t\right)}\,,
\end{equation}
where $\rho_1$ and $\rho_2$ are real constants, the MT model (\ref{MTa})--(\ref{MTb}) is transformed into the following bilinear equations
\begin{eqnarray}
&& \displaystyle (\mathrm{i} D_{x}- \frac{\rho_2}{\rho_1}) g \cdot f= - \frac{\rho_2}{\rho_1} h f^{\ast}  \,,  \label{MTdkBL1} \\ [5pt]
&& \displaystyle (\mathrm{i}D_{x}  - \sigma \rho^2_2 ) f \cdot f^{\ast } =-\sigma \rho^2_2 hh^{\ast} \,, \label{MTdkBL2}  \\ [5pt]
&& \displaystyle (\mathrm{i} D_{t} -\frac{\rho_1}{\rho_2}) h \cdot f^{\ast}  = - \frac{\rho_1}{\rho_2}g f  \,,  \label{MTdkBL3} \\ [5pt]
&& \displaystyle  ( \mathrm{i} D_{t}  - \sigma \rho^2_1 ) f^{\ast} \cdot f  =-\sigma \rho^2_1 gg^{\ast} \,. \label{MTdkBL4}
\end{eqnarray}
\end{prop}

\begin{proof}
By rewriting the dependent variable transformations (\ref{var_tran1})
\begin{equation*}
u=\rho_1 \frac{g}{f}   \frac{f}{f^{\ast} }  e^{\mathrm{i}(1+\sigma\rho_1\rho_2) \left(\frac{\rho_2}{\rho_1}x+\frac{\rho_1}{\rho_2}t\right)}\,, \quad
v=\rho_2 \frac{{h}}{f^{\ast}} \frac{f^{\ast}}{{f}} e^{\mathrm{i}(1+\sigma\rho_1\rho_2) \left(\frac{\rho_2}{\rho_1}x+\frac{\rho_1}{\rho_2}t\right)}
\end{equation*}
and substituting into Eq.(\ref{MTa}), one has
 \begin{eqnarray}
&& \left[ \mathrm{i}
\left(\frac{g}{f}\right)_x \frac{f}{f^{\ast} }  - \frac{\rho_2}{\rho_1}\frac{g}{f^{\ast}}  + \frac{\rho_2}{\rho_1} \frac{h}{f} \right] \nonumber \\
&&+ \left[ \mathrm{i}  \frac{g}{f} \left( \frac{f}{f^{\ast} }\right)_x -  \sigma\rho_2^2 \frac{g}{f^{\ast} }
 +  \sigma \rho_2^2 \frac{hh^*}{ff^*} * \frac{g}{f^{\ast} } \right] =0\,.
\end{eqnarray}
Bilinear equations (\ref{MTdkBL1}) and  (\ref{MTdkBL2}) are deduced by taking zero for each group inside bracket. Similarly, we can drive bilinear equations
(\ref{MTdkBL3}) and  (\ref{MTdkBL4}) by substituting (\ref{var_tran1})  into Eq.(\ref{MTb}).
\end{proof}

\subsection{Discrete KP equation and bilinear equations for the KP-Toda hierarchy}
Let us start with a concrete form of the Gram determinant expression of the
tau functions for the extended KP hierarchy with negative flows
\begin{equation}
\tau _{nkl}=\left\vert m_{ij}^{nkl}\right\vert _{1\leq i,j\leq N},
\label{KP-tau}
\end{equation}%
where
\begin{equation*}
m_{ij}^{nkl}=\delta _{ij}+\frac{\mathrm{i}p_{i}}{p_{i}+\bar{p}_{j}}\varphi _{i}^{nkl}\psi_{j}^{nkl},
\end{equation*}%
\begin{equation*}
\varphi _{i}^{nkl}=p_{i}^{n}(p_{i}-a)^{k}(p_{i}-b)^{l}e^{\xi _{i}},\quad
\psi
_{j}^{nkl}=\left(-\frac{1}{\bar{p}_{j}}\right)^{n}\left(-\frac{1}{\bar{p}_{j}+a}\right)^{k}
\left(-\frac{1}{\bar{p}_{j}+b}\right)^{l}e^{\bar{\xi}_{j}},
\end{equation*}
with
\begin{equation*}
\xi _{i}=\frac{1}{p_{i}}x_{-1}+p_{i}x_{1}+\frac{1}{p_{i}-a}t_{a}+\frac{1}{p_{i}-b}t_{b}+\xi _{i0},
\end{equation*}%
\begin{equation*}
\bar{\xi}_{j}=\frac{1}{\bar{p}_{j}}x_{-1}+\bar{p}_{j}x_{1}+\frac{1}{\bar{p}%
_{j}+a}t_{a}+\frac{1}{\bar{p}_{j}+b}t_{b}+\bar{\xi}_{j0}.
\end{equation*}%
Here $p_{i}$, $\bar{p}_{j}$, $\xi _{i0}$, $\bar{\xi}_{j0}$, $a$, $b$ are
constants.  We have the following lemma regarding the bilinear equations satisfied by above tau function:
\begin{lem}
The discrete KP equation generates a set of bilinear equations
\begin{eqnarray}
\label{dark-before-eq1}&&
(D_{x_1}+a) \tau_{n,k+1,l}\cdot \tau_{n+1,k,l} =a \tau_{n+1,k+1,l} \tau_{n,k,l}, \label{KPbilinear1} \\
\label{dark-before-eq2}&&
(bD_{x_{-1}}+1) \tau_{n,k,l+1}\cdot \tau_{n,k,l} = \tau_{n-1,k,l+1} \tau_{n+1,k,l}, \label{KPbilinear2} \\
\label{dark-before-eq3}&&
(aD_{t_{a}}-1)\tau_{n+1,k,l}\cdot \tau_{n,k,l}=-\tau_{n+1,k-1,l}\tau_{n,k+1,l}, \label{KPbilinear3} \\
\label{dark-before-eq4}&&
(bD_{t_{b}}-1)\tau_{n+1,k,l}\cdot \tau_{n,k,l}=-\tau_{n+1,k,l-1}\tau _{n,k,l+1}\,. \label{KPbilinear4}
\end{eqnarray}
satisfied by above tau function (\ref{KP-tau}).
\end{lem}
\begin{proof}
The discrete KP equation, or the so-called Hirota--Miwa equation \cite{Hirota-1981, Miwa-1982},
\begin{equation}
\label{H-M-1}
(a_1-a_2)\tau_{12}\tau_3+(a_2-a_3)\tau_{23}\tau_1+(a_3-a_1)\tau_{13}\tau_2=0,
\end{equation}%
is a three-dimensional discrete integrable system where lattice parameters $a_k$ are distinct constants, $k=1, 2, 3$, and for $\tau=\tau(k_1,k_2,k_3)$ each subscript $i$ denotes a forward shift in the corresponding discrete variable $k_i$.
It is found by Ohta {\it et al.} that the discrete KP equation admits a general solution in terms of the following Gram-type  determinant \cite{OHTI93}
\begin{equation}
\tau (k_{1},k_{2},k_{3})=\Big|c_{ij}+\frac{d_{ij}}{p_{i}+q_{j}} \left(- \frac{p_{i}-a_{1}}{q_{j}+a_{1}}\right) ^{k_{1}} \left(- \frac{p_{i}-a_{2}}{q_{j}+a_{2}}\right) ^{k_{2}} \left(-\frac{p_{i}-a_{3}}{q_{j}+a_{3}}\right) ^{k_{3}}\Big|\,.
\label{HW-Gram}
\end{equation}
Notice that the element in  (\ref{HW-Gram}) can be rewritten as
\begin{eqnarray*}
&& c_{ij} + \frac{d_{ij}}{p_{i}+q_{j}} \left( -\frac{\widetilde{p}_{i}}{\widetilde{q}_{j}}\right) ^{k_{1}}
  \left( -\frac{\widetilde{p}_{i}+a}{%
\widetilde{q}_{j}-a}\right) ^{k_{3}} \left( \frac{1-b\widetilde{p}_{i}}{%
1+b\widetilde{q}_{j}}\right) ^{k_{2}}  \left( \frac{1-a_{3}p^{-1}_{i}}{%
1+a_{3}q^{-1}_{j}}\right) ^{k_{3}}
\end{eqnarray*}
by reparametrizing ${p}_i-a_1=\widetilde{p}_{i}$, ${q}_i+a_1=\widetilde{q}_{i}$ and $a_2-a_1=b^{-1}$ and $a=a_{1}$.
Set $a_3=0$ and redefine  $k_{3}=k$, $k_1=n$, then the discrete KP equation (\ref{H-M-1}) has the degenerate form
\begin{eqnarray}
&& a\tau_{n,k}(k_2+1) \tau _{n+1,k+1}(k_{2}) +b^{-1} \tau _{n,k+1}(k_2)\tau _{n+1,k}(k_2+1)  \nonumber \\
&& (b^{-1}+a) \tau_{n,k+1}(k_2+1) \tau _{n+1,k+1}(k_{2})=0.
\label{full-discrete}
\end{eqnarray}
Applying Miwa transformation by taking  $b \to 0$ and $x_{1} =-k_2b$, i.e., $\tau_{n,k}(k_2+1) \to \tau_{n,k}- b \partial_{x_1} \tau_{n,k}$
one obtains
\begin{equation}
(D_{x_1}+a)  \tau_{n+1,k} \cdot \tau_{n,k+1} =a \tau_{n+1,k+1} \tau_{n,k},
\end{equation}
which is equivalent to  (\ref{KPbilinear1}) by taking $\widetilde{p}_{i} \to p_{i}$, $\widetilde{q}_{i} \to \bar{p}_{i}$, $c_{ij}=\delta_{ij}$, $d_{ij}= \mathrm{i}p_{i}$, $a \to -a$ and adding $l$ to each tau function.

In what follows, we further show that bilinear equation (\ref{KPbilinear1})  can generate  (\ref{KPbilinear2})--(\ref{KPbilinear4}) by a dual relation between positive flow and negative flow. To be specific, we notice that
\begin{eqnarray*}
&&\delta_{ij}+ \frac{\mathrm{i}p_{i}}{p_{i}+\bar{p}_{j}} \left(-\frac{{p}_{i}}{\bar{p}_{j}}\right)^{n}\left(-\frac{{p}_{i}-a}{\bar{p}_{j}+a}\right)^{k} e^{\xi_i+\bar{\xi}_j} \\
&&\to  \delta_{ij}+ \frac{\mathrm{i}p_{i}}{p_{i}+\bar{p}_{j}} \left(-\frac{{p}^{-1}_{i}}{\bar{p}^{-1}_{j}}\right)^{-(n+k)}\left(-\frac{{p}^{-1}_{i}-a^{-1}}{\bar{p}^{-1}_{j}+a^{-1}}\right)^{k} e^{\xi_i+\bar{\xi}_j}\,.
\end{eqnarray*}
Therefore, $x_1$ and $x_{-1}$ is exchangeable by redefining ${p}^{-1}_{i} \to {p}_{i}$, $\bar{p}^{-1}_{j} \to \bar{p}_{j}$. Furthermore, by redefining index $n+k \to -n$, $k \to l$ and  $a^{-1} = b$,  (\ref{KPbilinear1}) is converted into
\begin{eqnarray*}
&&
(D_{x_{-1}}+b^{-1})  \tau_{n,l+1}\cdot \tau_{n,l} =b^{-1} \tau_{n-1,l} \tau_{n+1,l}
\end{eqnarray*}
which is nothing but Eq. (\ref{KPbilinear2}). 

On the other hand,  by reparametrizing  $p^{-1}_{i} -a^{-1}=\widetilde{p}_{i}$,  $\bar{p}^{-1}_{j} +a^{-1}=\widetilde{\bar{p}}_{j}$
\begin{eqnarray*}
&& c_{ij}+ \frac{d_{ij}}{p_{i}+\bar{p}_{j}}  \left(-\frac{{p}_{i}}{\bar{p}_{j}}\right)^{n}\left(-\frac{{p}_{i}-a}{\bar{p}_{j}+a}\right)^{k} e^{\xi_i+\bar{\xi}_j} \\
&&\to  c_{ij}+ \frac{d_{ij}}{p_{i} +\bar{p}_{j}}  \left(-\frac{\widetilde{p}_{i}}{\widetilde{q}_{j}}\right)^{k}\left(-\frac{\widetilde{p}_{i}+a^{-1}}{\widetilde{q}_{j}-a^{-1}}\right)^{-(n+k)}
e^{\xi_i+\bar{\xi}_j}  
\end{eqnarray*}
by redefining  indices $k=n'$, $n+k=-k'$. Since
\begin{eqnarray*}
&&
\xi_i = p_i x_1 +  p^{-1}_i x_{-1} \to \frac{1}{\widetilde{p}_{i}+a^{-1}} x_1 + (\widetilde{p}_{i}+a^{-1}) x_{-1}, \\
&&\bar{\xi}_j = \bar{p_j} x_1 +  \bar{p}_j^{-1} x_{-1} \to \frac{1}{\widetilde{\bar{p}}_{j}-a^{-1}} x_1 + (\widetilde{\bar{p}}_{i}-a^{-1}) x_{-1}\,,
\end{eqnarray*}
which makes positive flow  ($x_1$) and negative flow ($x_{-1}$) exchangeable. Thus, by taking $x_1 \to t_a$,  and $a^{-1} \to -a$, one obtains
\begin{eqnarray*}
&&(D_{x_1}  + a) \tau _{n,k+1}\cdot \tau _{n+1,k}=a \tau _{n,k} \tau _{n+1,k+1} \\
&& \to (D_{t_a}  - a^{-1}) \tau _{n'+1,k'}\cdot \tau _{n',k'}= -a^{-1} \tau _{n'+1,k-1} \tau _{n',k'+1}  \\
&& \to (aD_{t_{a}}-1)\tau _{n'+1,k'}\cdot \tau _{n',k'}=-\tau
_{n'+1,k'-1}\tau _{n',k'+1}
\end{eqnarray*}
which is exactly Eq. (\ref{KPbilinear3}) by dropping the prime.  Eq. (\ref{KPbilinear4}) is just a parallel copy of Eq. (\ref{KPbilinear3}) from  $(a, t_a, k)  \to (b,t_b,l)$.
\end{proof}
\begin{remark} It is very interesting to observe that the discrete KP equation can generate KP-Toda hierarchy with asymmetric positive flow and negative flow. Furthermore, the positive flow and negative flow are exchangeable by reparametrizing the wave numbers in the tau function. 
\end{remark}
\begin{remark} It is noted that some of the bilinear equations derived above are also bilinear equations of the Fokas-Lenells equation \cite{MatsunoFL2}, the complex short pulse equation 
\cite{FMO_ComplexSPE} and the modified Camassa-Holm equation \cite{Matsuno-mCH,ZYFmCH-discrete}. In other words, the tau function behind these equations is the same before the reductions. 
\end{remark}
\begin{remark} We can also give proof of above bilinear equations via determinant identities such as the Jacobi identity. However, the proof we give here is more systematic instead of technical. Moreover, in scrutinizing the connection between the discrete KP equation and the nonlinear PDEs, it makes possible for us to construct the integrable discrete analogues of the nonlinear PDEs. 
\end{remark}
\subsection{Reduction to the dark soliton of the MT model}
In what follows, we briefly show the reduction processes of reducing bilinear
equations of extended KP hierarchy (\ref{KPbilinear1})--(\ref{KPbilinear4}) to the bilinear equation (\ref{MTdkBL1})--(\ref{MTdkBL4}).
Firstly, we start with dimension reduction by noting that the determinant expression of $\tau _{nkl}$,
\begin{equation*}
\tau _{nkl}=\left\vert \delta _{ij}+\frac{\mathrm{i}p_{i}}{p_{i}+\bar{p}_{j}}\varphi
_{i}^{nkl}\psi _{j}^{nkl}\right\vert _{1\leq i,j\leq N}\,,
\end{equation*}
can be alternatively expressed by%
\begin{equation*}
\tau _{nkl}=\prod_{i=1}^N \varphi _{i}^{nkl} \psi _{i}^{nkl} \left\vert \frac{\delta _{ij}}{ \varphi_{i}^{nkl}\psi _{i}^{nkl}}+\frac{\mathrm{i}p_{i}}{p_{i}+\bar{p}_{j}} \right\vert _{1\leq i,j\leq N}\,,
\end{equation*}
by dividing $j$-th column by $\psi _{j}^{nkl}$ and  $i$-th row by $%
\varphi _{i}^{nkl}$ for $1\leq i,j\leq N$.

By imposing the reduction condition
\begin{equation}\label{condition-1}
(p_i-b)(\bar{p}_i+b)=b(a-b),
\end{equation}
or
\begin{equation}
|p_i|^2=-\frac{b}{a-b}(p_i-a)(\bar{p}_i+a),
\end{equation}
one can easily check the following relations hold
\begin{eqnarray}
&& p_i+\bar{p}_i=b(a-b) \left( \frac{1}{p_i-b} + \frac{1}{\bar{p}_i+b} \right),\\
&& \frac{1}{p_i}+\frac{1}{\bar{p}_i}=-\frac{a-b}{b} \left( \frac{1}{p_i-a} + \frac{1}{\bar{p}_i+a} \right),\\
&& \left( -\frac{p_{i}}{\overline{p}_{i}}\right)
\left( -\frac{p_{i}-a}{\overline{p}_{i}+a}\right)
\left( -\frac{p_{i}-b}{\overline{p}_{i}+b}\right) ^{-1}=1,
\end{eqnarray}
which implies
\begin{eqnarray}
&& \partial_{t_b}=\frac{1}{b(a-b)}\partial_{x_1},   \\
&& \partial_{t_a}=-\frac{b}{a-b}\partial_{x_{-1}},\\
&& \tau_{n-1,k,l+1}= \tau_{n,k+1,l},\ \ \tau_{n,k,l+1}= \tau_{n+1,k+1,l}\,.
\end{eqnarray}
Therefore the bilinear equations (\ref{KPbilinear1})--(\ref{KPbilinear4}) can be recast into
\begin{eqnarray}
&& (D_{x_1}+a) \tau_{n,k+1,l}\cdot \tau_{n+1,k,l} =a \tau_{n,k,l+1} \tau_{n,k,l}, \label{dark-before-eeq1}\\
&& (bD_{x_{-1}}+1) \tau_{n,k,l+1}\cdot \tau_{n,k,l} = \tau_{n,k+1,l} \tau_{n+1,k,l},  \label{dark-before-eeq3}\\
&& (-\frac{ab}{(a-b)}D_{x_{-1}}-1)\tau_{n+1,k,l}\cdot \tau_{n,k,l}=-\tau_{n+1,k-1,l}\tau_{n,k+1,l},\\
&& (\frac{1}{a-b}D_{x_{1}}-1)\tau_{n+1,k,l}\cdot \tau_{n,k,l}=-\tau_{n+1,k,l-1}\tau _{n,k,l+1}\,. \label{dark-before-eeq4}
\end{eqnarray}
By setting
\begin{equation}
x_1=-\frac{\rho_2}{{\rm i}a\rho_1}x,\ \ x_{-1}=-\frac{\rho_1 b}{{\rm i}\rho_2}t,\ \
\end{equation}
i.e.,
\begin{equation}
\partial_{x_1}= -\frac{{\rm i}a\rho_1}{\rho_2}  \partial_{x},\ \ \partial_{x_{-1}}=-\frac{{\rm i}\rho_2}{\rho_1 b} \partial_{t},
\end{equation}
and assuming $b=a(1+\sigma \rho_1\rho_2)$, we have the following  bilinear equations
\begin{eqnarray}
\label{dark-before-eeeeq1}&&
( {\rm i}D_{x}-\frac{\rho_2}{\rho_1}) \tau_{n,k+1,l}\cdot \tau_{n+1,k,l} =-\frac{\rho_2}{\rho_1}\tau_{n,k,l+1} \tau_{n,k,l},\\
\label{dark-before-eeeeq2}&&
({\rm i} D_{t} - \frac{\rho_1}{\rho_2}) \tau_{n,k,l+1}\cdot \tau_{n,k,l} = - \frac{\rho_1}{\rho_2}\tau_{n,k+1,l} \tau_{n+1,k,l},\\
\label{dark-before-eeeeq3}&&
( -{\rm i}D_{t}-\sigma\rho^2_1)\tau_{n+1,k,l}\cdot \tau_{n,k,l}=-\sigma\rho^2_1\tau_{n+1,k-1,l}\tau_{n,k+1,l},\\
\label{dark-before-eeeeq4}&&
({\rm i}D_{x}-\sigma\rho^2_2)\tau_{n+1,k,l}\cdot \tau_{n,k,l}=-\sigma\rho^2_2\tau_{n+1,k,l-1}\tau _{n,k,l+1},
\end{eqnarray}

Next, we proceed to the complex conjugate reduction. To this end, by taking $a$ and $b$ pure imaginary,
and letting $\bar{p_i}$ to be complex conjugate of $p_i$: $\bar{p_i}=p^*_i$ and $\bar{\xi}_{i0}=\xi^*_{i0}$ and $n=-1,k=0,l=0$, one can find that
\begin{equation}
\tau_{-1,0,0}=\tau^*_{0,0,0},\ \ \tau_{-1,1,0}=\tau^*_{0,-1,0},\ \ \tau_{-1,0,1}=\tau^*_{0,0,-1}
\end{equation}
In summary, by defining
\begin{equation*}
\tau_{0,0,0,}=f, \ \tau_{-1,0,0}=f^*, \  \tau_{-1,1,0}=g, \  \tau_{0,-1,0}=g^*, \  \tau_{-1,0,1}=h, \  \tau_{0,0,-1}=h^*
\end{equation*}%
we arrive at exactly the set of bilinear equations (\ref{KPbilinear1})--(\ref{KPbilinear4}). Therefore, the reduction process is complete. As a result, we can provide the determinant solution to the MT model  by the following theorem.
In summary, we can give the multi-dark soliton solution to the MT system by the following theorem by taking $a={\rm i}\alpha$ and $b={\rm i}\alpha(1+\sigma\rho_1\rho_2)$.
\begin{theorem}
The MT system  (\ref{MTa})--(\ref{MTb}) admits the multi-dark soliton solution
$$
u=\rho_1 \frac{g}{f^{\ast} } e^{\mathrm{i}(1+\sigma\rho_1\rho_2) \left(\frac{\rho_2}{\rho_1}x+\frac{\rho_1}{\rho_2}t\right)}\,, \quad
v=\rho_2 \frac{{h}}{{f}} e^{\mathrm{i}(1+\sigma\rho_1\rho_2) \left(\frac{\rho_2}{\rho_1}x+\frac{\rho_1}{\rho_2}t\right)}\,,
$$
where $f$, $f^*$, $g$ and $h$ are the following $N \times N$ determinant
\begin{eqnarray}
\label{multi-dark-sol-theorem-01}&&
f=\left\vert \delta _{ij}+\frac{\mathrm{i}p_{i}}{p_{i}+{p}^*_{j}}
e^{\xi _{i}+\bar{\xi}_{j}} \right\vert \,, \ \
f^*=\left\vert \delta _{ij}-\frac{\mathrm{i}{p}^*_{j}}{p_{i}+{p}^*_{j}}
e^{\xi _{i}+{\xi}^*_{j}} \right\vert \,,
\\
\label{multi-dark-sol-theorem-02}&&
g=\left\vert \delta _{ij}+\frac{-\mathrm{i}p^*_{j}}{p_{i}+{p}^*_{j}}
\left(- \frac{p_{i}-{\rm i}\alpha}{{p}^*_{j}+{\rm i}\alpha}\right)
e^{\xi _{i}+{\xi}^*_{j}} \right\vert \,,
\\
\label{multi-dark-sol-theorem-03}&&
h=\left\vert \delta _{ij}+\frac{-\mathrm{i}p^*_{j}}{p_{i}+{p}^*_{j}}
\left[- \frac{p_{i}-{\rm i}\alpha(1+\sigma\rho_1\rho_2)}{{p}^*_{j}+{\rm i}\alpha(1+\sigma\rho_1\rho_2)}\right]
e^{\xi _{i}+{\xi}^*_{j}} \right\vert \,,
\end{eqnarray}
with
\begin{eqnarray*}
\xi _{i}=\frac{\rho_2}{\alpha\rho_1} p_{i}x -\frac{\rho_1  \alpha (1+\sigma\rho_1\rho_2) }{ \rho_2}\frac{t}{p_{i}} +\xi_{i0}.
\end{eqnarray*}
Here $p_i$, $\xi_{i0}$ are complex constants and $\alpha$ is a real constant which need to satisfy the constraint condition:
\begin{equation}\label{constraint-condition}
[p_i-{\rm i}\alpha(1+\sigma\rho_1\rho_2)][p^*_i+{\rm i}\alpha(1+\sigma\rho_1\rho_2)]=\alpha^2\rho_1\rho_2(\sigma+\rho_1\rho_2)\,.
\end{equation}
\end{theorem}

\section{Dynamics of bright and dark soliton solutions}
\subsection{One- and two-bright solitons and their dynamics}
Based on the $N$-bright soliton solution to the MT model, one- and
two-soliton solutions are calculated as follows.
For simplicity, we reparameterize $\xi'_i=\xi_i+\tilde{\xi}'_i$ with $\alpha_i=e^{\tilde{\xi}'_i}$,
and redefine the complex constants $p_i$ and $\xi'_i$ as
\begin{equation}\label{reparemeterrize01}
p_i=a_i+{\rm i}b_i,\ \ \xi'_i=\eta_{i0}+{\rm i}\theta_{i0},\ \ i=1,2,
\end{equation}
where $a_i$, $b_i$, $\eta_{i0}$ and $\theta_{i0}$ are real constants.
Then the variables $\xi'_1$ and $\xi'_2$ are rewritten as
\begin{equation}\label{reparemeterrize02}
\xi'_i=\eta_i+{\rm i}\theta_i,\ \
\eta_i=a_i(x- \frac{1}{a^2_i+b^2_i}t)+\eta_{i0},\ \
\theta_i=b_i(x+\frac{1}{a^2_i+b^2_i}t)+\theta_{i0},\ \ i=1,2.
\end{equation}

By taking $N=1$
in (\ref{multi-bright-sol-theorem-01})-(\ref{multi-bright-sol-theorem-02}),
we obtain the tau functions for one-soliton solution
\begin{eqnarray}
f =1+\frac{\mathrm{i}\sigma p_1 |\alpha _1|^{2}} {(p_1+p^*_1)^2 }e^{\xi _{1}+\xi^* _{1}} ,
\ \
g =\frac{\mathrm{i} \alpha_1^*} {p_1} e^{\xi _{1}},
\ \
h=\alpha_1^* e^{\xi _{1}},
\end{eqnarray}
or
\begin{eqnarray}\label{mt-bright-one-soliton1}
f =1+\frac{\mathrm{i} \sigma (a_1+{\rm i} b_1) } {4a^2_1 }e^{2\eta_1} ,
\ \
g =\frac{\mathrm{i}} {a_1+{\rm i} b_1} e^{\eta _{1}+{\rm i}\theta_1},
\ \
h=e^{\eta _{1}+{\rm i}\theta_1}.
\end{eqnarray}
This leads to the square of the modulus of one-soliton solution for the MT model (\ref{MTa}) and (\ref{MTb})
\begin{eqnarray}
&& |u|^2=\frac{2a^2_1}{(a^2_1+b^2_1)^{\frac{3}{2}}}\frac{1}{{\rm cosh}(2\eta_1+2\delta)-\frac{\sigma b_1}{\sqrt{a^2_1+b^2_1}}},\\
&& |v|^2=\frac{2a^2_1}{(a^2_1+b^2_1)^{\frac{1}{2}}}\frac{1}{{\rm cosh}(2\eta_1+2\delta)-\frac{\sigma b_1}{\sqrt{a^2_1+b^2_1}}},
\end{eqnarray}
with $e^{2\delta}=\frac{\sqrt{a^2_1+b^2_1}}{4a^2_1}$.
Then the amplitudes $A_u$ and $A_v$ are given by
\begin{eqnarray}
&& A_u=\sqrt{\frac{2}{a^2_1+b^2_1}(\sqrt{a^2_1+b^2_1}+\sigma b_1)}=\sqrt{2(v^{\frac{1}{2}}_1+\sigma b_1v_1)},\ \
\\
&&
A_v=\sqrt{2(\sqrt{a^2_1+b^2_1}+\sigma b_1)}=\sqrt{2(v^{-\frac{1}{2}}_1+\sigma b_1)},
\end{eqnarray}
which means the amplitude-velocity relations with the velocity $\kappa_1=\frac{1}{a^2_1+b^2_1}$.
For a fixed value of the imaginary part of wave number $b_1$, it can be seen that when $\sigma=1$, $A_u$ $(A_v)$ is an increasing (decreasing) function of $\kappa_1$ , while when $\sigma=-1$, $A_u$ increases in the interval $0<\kappa_1\leqslant \frac{1}{4b^2_1}$ and decreases in the interval $\frac{1}{4b^2_1}\leqslant \kappa_1 <\frac{1}{b^2_1}$, $A_v$ is a decreasing function in the interval $0 < \kappa_1 <\frac{1}{b^2_1}$.
Two cases of the amplitude-velocity relations are depicted in Fig. \ref{fig1}.

\begin{figure}[!htbp]
\centering
{\includegraphics[height=1.6in,width=4.0in]{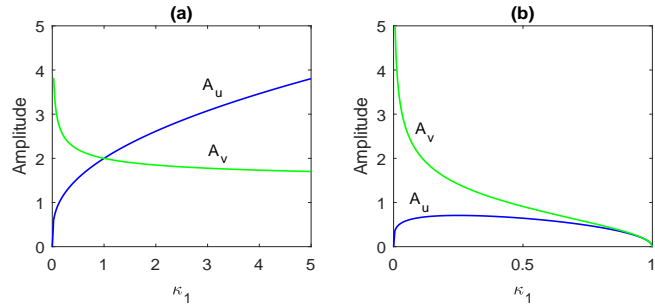}}
\caption{The amplitude-velocity relations with $b_1=1$: (a) $\sigma=1$ and (b) $\sigma=-1$.
\label{fig1}}
\end{figure}

Besides, we need to point out that the amplitudes $A_u$ and $A_v$ remain finite in the limit of $a_1\rightarrow 0$.
In this limit case, the bright one-soliton solution (\ref{mt-bright-one-soliton1}) degenerates to the following algebraic soliton
solution
\begin{equation}\label{algebraic-soliton1}
u=\pm\frac{2\sqrt{\frac{\sigma}{b_1}}e^{{\rm i}\left[b_1(x+\frac{1}{b^2_1}t)+\theta'_{10}\right]}}{b_1[2(x-\frac{1}{b^2_1}t+x_0)+\frac{{\rm i}}{b_1}]},
\ \
v=\pm\frac{2\sqrt{\frac{\sigma}{b_1}}e^{{\rm i}\left[b_1(x+\frac{1}{b^2_1}t)+\theta'_{10}\right]}}{2(x-\frac{1}{b^2_1}t+x_0)+\frac{{\rm i}}{b_1}},
\end{equation}
where we take $e^{\eta_{10}+{\rm i}\theta_{10}}=\mp 2a_1\sqrt{\frac{\sigma}{b_1}}e^{a_1x_0+{\rm i}\theta'_{10}}$.
The profiles of one-soliton are plotted in Fig. \ref{fig2} with different values of wave number $p_1$.

\begin{figure}[!htbp]
\centering
{\includegraphics[height=1.6in,width=4.0in]{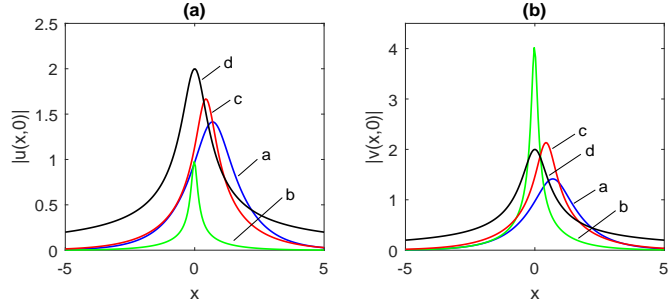}}
\caption{The bright one-soliton  with the parameters $\sigma=1$ and $\xi_{10}=x_0=0$: (a) $p_1=1$, (b) $p_1=1+4{\rm i}$, (c) $p_1=0.8+{\rm i}$ and (d) the algebraic soliton with $p_1={\rm i}$.
\label{fig2}}
\end{figure}

The tau functions corresponding to two-bright soliton solution can be obtained  by taking $N=2$ in (\ref{multi-bright-sol-theorem-01})-(\ref{multi-bright-sol-theorem-02})
\begin{eqnarray}
\nonumber &&
f=1+c_{1 1^*}e^{\xi _{1}+{\xi}^*_{1}}+c_{21^*}e^{\xi _{2}+{\xi}^*_{1}} +c_{12^*}e^{\xi _{1}+{\xi}^*_{2}}
\\
&&\hspace{1cm}
+c_{22^*}e^{\xi _{2}+{\xi}^*_{2}} +c_{121^*2^*}e^{\xi _{1}+\xi _{2}+{\xi}^*_{1}+{\xi}^*_{2}}\,,
\\
&&
g =  \frac{\mathrm{i}\alpha^{\ast}_{1}}{p_1}e^{\xi _{1}} + \frac{\mathrm{i}\alpha^{\ast}_{2}}{p_2}e^{\xi_{2}}-\frac{\mathrm{i}p^*_1}{p_1p_2} c_{121^*}e^{\xi_{1}+\xi_{2}+{\xi}^*_{1}}-\frac{\mathrm{i}p^*_2}{p_1p_2} c_{122^*}e^{\xi _{1}+\xi_{2}+{\xi}^*_{2}}\,,
\\
&&
h=\alpha^*_{1}e^{\xi _{1}}+\alpha^*_{2}e^{\xi _{2}}+c_{121^*}e^{\xi_{1}+\xi_{2}+{\xi}^*_{1}}+c_{122^*}e^{\xi _{1}+\xi_{2}+{\xi}^*_{2}}\,,
\end{eqnarray}%
where%
\begin{eqnarray*}
&&
c_{ij^*}=\frac{{\rm i}\sigma{p}_{i} \alpha^*_i\alpha_j}{({p}_{i}+{p}_{j}^{\ast
})^{2}},\quad
c_{12i^*}=\left( p_1-p_{2}\right) p^*_i \left[ \frac{\alpha^*_{2}c_{1i^*}}{p_1(p_{2}+{p}_{i}^{\ast })}-\frac{\alpha^*_{1}c_{2i^*}}{p_2(p_{1}+{p}_{i}^{\ast })}\right] \,,
\\
&&
c_{121^*2^*}=|p_1-p_2|^2\left[ \frac{c_{11^*}c_{22^*}}{%
(p_{1}+{p}_{2}^{\ast })\left( p_{2}+{p}_{1}^{\ast }\right) }-\frac{c_{12^*}c_{21^*}}{(p_{1}+{p}_{1}^{\ast })\left( p_{2}+{p}_{2}^{\ast }\right) }%
\right] \,.
\end{eqnarray*}
The interaction of two-bright solitons are illustrated in Fig.\ref{fig3} for different parameters.
As shown in Fig. \ref{fig3}(a) and (b), two-bright solitons undergo the regular collisions.
The bound soliton state is exhibited in Fig. \ref{fig3}(c), in which two solitons move with the same velocity $\frac{1}{a^2_1+b^2_1}=\frac{1}{a^2_2+b^2_2}$.

\begin{figure}[!htbp]
\centering
{\includegraphics[height=2.6in,width=5.0in]{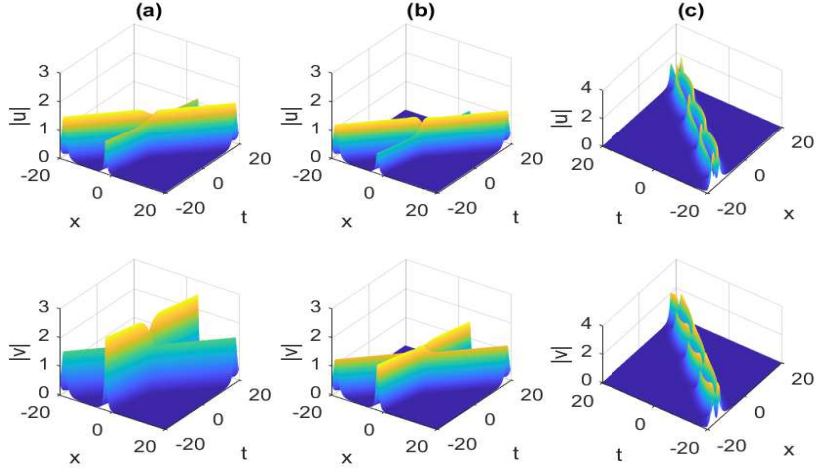}}
\caption{The bright two-soliton  with the parameters $\alpha_1=\alpha_2=1$, $\xi_{10}=\xi_{20}=0$, $p_1=1+\frac{1}{5}{\rm i}$ and $p_2=2+{\rm i}$: (a) $\sigma=1$ and (b) $\sigma_1=-1$; (c) the bound state with $\sigma=1$, $p_1=1+\frac{2}{5}{\rm i}$ and $p_2=\frac{4}{5}+\frac{\sqrt{13}}{5}{\rm i}$.
\label{fig3}}
\end{figure}

\subsection{One- and two-dark solitons and their dynamics}
In this subsection, we list one- and two-dark soliton solutions to the MT model.
Same as the bright soliton case, we rewrite the complex constants $p_i$ and $\xi_i$ as
\begin{equation}\label{dark-reparemeterrize01}
p_i=a_i+{\rm i}b_i,\ \ \xi_i=\eta_{i0}+{\rm i}\theta_{i0},\ \ i=1,2,
\end{equation}
where $a_i$, $b_i$, $\eta_{i0}$ and $\theta_{i0}$ are real constants.
Then the variables $\xi_1$ and $\xi_2$ are rewritten as
\begin{eqnarray}\label{dark-reparemeterrize02}
\nonumber &&
\xi_i=\eta_i+{\rm i}\theta_i,\ \ \kappa_i= \frac{\alpha^2\rho^2_1(1+\sigma\rho_1\rho_2)}{\rho^2_2(a^2_i+b^2_i)},
\\
&&
\eta_i=\frac{\rho_2}{\alpha\rho_1}a_i (x- \kappa_i t)+\eta_{i0},\ \
\theta_i=\frac{\rho_2}{\alpha\rho_1}b_i (x+\kappa_i t)+\theta_{i0},\ \
i=1,2.
\end{eqnarray}

By taking $N=1$, we have tau functions for one-dark soliton solution
\begin{eqnarray}
&& f=1+\frac{\mathrm{i} p_1 } {(p_1+p^*_1)}e^{\xi _{1}+\xi^* _{1}}\,,
\ \
g=1+\frac{\mathrm{i} p^*_1 } {(p_1+p^*_1)} \left( \frac{p_1-\mathrm{i} \alpha}{p^*_1+\mathrm{i} \alpha} \right) e^{\xi _{1}+\xi^* _{1}}\,,
\\
&&
h=1+\frac{\mathrm{i} p^*_1} {(p_1+p^*_1)} \left[ \frac{p_1- {\rm i}\alpha(1+\sigma\rho_1\rho_2)}{p^*_1+{\rm i}\alpha(1+\sigma\rho_1\rho_2)} \right] e^{\xi _{1}+\xi^* _{1}}\,,
\end{eqnarray}
which leads to the squares of the modulus of $u$ and $v$:
\begin{eqnarray}
&&
|u|^2=\rho^2_1\left[1+\frac{2{\rm sgn}(a_1)\alpha a^2_1}{[a^2_1+(b_1-\alpha)^2]\sqrt{a^2_1+b^2_1}}\frac{1}{{\rm cosh}(2\eta_1+2\delta')-\frac{{\rm sgn}(a_1) b_1}{\sqrt{a^2_1+b^2_1}}} \right],\ \ \ \ \ \
\\
&&
|v|^2=\rho^2_2\left[1+\frac{2{\rm sgn}(a_1)(\alpha+\sigma\alpha\rho_1\rho_2) a^2_1}{[a^2_1+(b_1-\alpha-\sigma\alpha\rho_1\rho_2)^2]\sqrt{a^2_1+b^2_1}}\frac{1}{{\rm cosh}(2\eta_1+2\delta')-\frac{{\rm sgn}(a_1) b_1}{\sqrt{a^2_1+b^2_1}}} \right],\ \ \ \ \ \
\end{eqnarray}
with $e^{4\delta'}=\frac{a^2_1+b^2_1}{4a^2_1}$.
This implies that $u$ exhibits a dark soliton when $\alpha a_1<0$ and an anti-dark soliton on the background $\rho_1$ when $\alpha a_1>0$, while $v$ represents a dark soliton when $(\alpha+\sigma\alpha\rho_1\rho_2) a_1<0$ and an anti-dark soliton on the background $\rho_1$ when $(\alpha+\sigma\alpha\rho_1\rho_2) a_1>0$.

Besides, the dark one-soliton solution can be rewritten as
\begin{eqnarray*}
&&
u=\frac{\rho_1}{2}  e^{\mathrm{i}(1+\sigma\rho_1\rho_2) \left(\frac{\rho_2}{\rho_1}x+\frac{\rho_1}{\rho_2}t\right)}
\left[ 1+e^{2{\rm i}\phi_1}+ (e^{2{\rm i}\phi_1}-1)\tanh(\eta_1+\eta_0+{\rm i}\phi_0) \right],\ \ \
\\
&&
v=\frac{\rho_2}{2}  e^{\mathrm{i}(1+\sigma\rho_1\rho_2) \left(\frac{\rho_2}{\rho_1}x+\frac{\rho_1}{\rho_2}t\right)}
\left[ 1+e^{4{\rm i}\phi_0+2{\rm i}\phi_2 }+ (e^{4{\rm i}\phi_0+2{\rm i}\phi_2}-1)\tanh(\eta_1+\eta_0-{\rm i}\phi_0) \right],\ \ \
\end{eqnarray*}
where
\begin{eqnarray*}
e^{2\eta_0+2{\rm i}\phi_0}=\frac{-{\rm i}p^*_1}{p_1+p^*_1},\ \
e^{2{\rm i}\phi_1}=-\frac{p_1-\mathrm{i} \alpha}{p^*_1+\mathrm{i} \alpha},\ \
e^{2{\rm i}\phi_2}=-\frac{p_1- {\rm i}\alpha(1+\sigma\rho_1\rho_2)}{p^*_1+{\rm i}\alpha(1+\sigma\rho_1\rho_2)},
\end{eqnarray*}
Therefore, the phase of $u$ and $v$ acquire shifts in the amount of $2\phi_1$ and $4\phi_0+2\phi_2$
when $\eta_1$ varies from $-\infty$ to $+\infty$, and the grayness of two components are $|\rho_1\cos\phi_1|$ and $|\rho_2\cos(2\phi_0+\phi_2)|$, respectively.

Moreover, the constraint condition (\ref{constraint-condition}) becomes
\begin{equation}
a^2_1+[b_1-\alpha(1+\sigma\rho_1\rho_2)]^2=\sigma\rho_1\rho_2(1+\sigma\rho_1\rho_2)\alpha^2,
\end{equation}
which implies that the condition $\sigma\rho_1\rho_2(1+\sigma\rho_1\rho_2)>0$ needs to be hold.
If the above constraint condition is expressed in terms of the velocity $\kappa_1$, one can find that
\begin{equation}
b_1=\frac{1}{2}\left(\alpha + \frac{\rho^2_1}{\kappa_1\rho^2_2}\right),\ \
a^2_1=\frac{\alpha^2}{4\kappa^2_1}(\kappa_{1,max}-\kappa_1)(\kappa_1-\kappa_{1,min}),\ \
\kappa_{1,min}<\kappa_1<\kappa_{1,max},
\end{equation}
where $\kappa_{1,min}=\hat{k}^{-}_1$, $\kappa_{1,max}=\hat{k}^{+}_1$ when $\alpha>0$, and $\kappa_{1,min}=\hat{k}^{+}_1$, $\kappa_{1,max}=\hat{k}^{-}_1$ when $\alpha<0$ with $\hat{k}^{\pm}_1=\frac{\rho^2_1}{\alpha\rho^2_2}[1+2\sigma\rho_1\rho_2 \pm 2\sqrt{\sigma\rho_1\rho_2(1+\sigma\rho_1\rho_2)}]$.
In the following, we discuss the maximum amplitude-velocity relations.
Without loss of generality, we consider $\alpha>0$, $1+\sigma\rho_1\rho_2>0$, $\rho_1,\rho_2>0$ which lead to $\kappa_1>0$ and $\sigma=1$,
then there are two cases corresponding to the sign of $a_1$.

\textbf{Case 1}: $a_1>0$. Two components exhibit as anti-dark soliton and the amplitude-velocity relations are given by
\begin{eqnarray}
&&
A_u=\sqrt{\rho^2_1+\frac{2\alpha\rho^2_1(\sqrt{a^2_1+b^2_1}+b_1)}{a^2_1+(b_1-\alpha)^2}  }-\rho_1
=\sqrt{\rho^2_1+ \Delta^{+}_1 } -\rho_1,
\\
&&
A_v=\sqrt{\rho^2_2+\frac{2\alpha(1+\rho_1\rho_2)\rho^2_2(\sqrt{a^2_1+b^2_1}+b_1)}{a^2_1+(b_1-\alpha-\alpha\rho_1\rho_2)^2}  }-\rho_2
=\sqrt{\rho^2_2+ \frac{\Delta^{+}_1}{\alpha\kappa_1} } -\rho_2,\ \ \ \ \ \
\end{eqnarray}
with $\Delta^{+}_1=\frac{\rho^2_1+\alpha\kappa_1\rho^2_2+2\sqrt{\alpha\kappa_1(1+\rho_1\rho_2)\rho^2_1\rho^2_2} }{\rho_1\rho_2}$.
$A_u$ is an increasing function whereas $A_v$ is a decreasing function in the interval $\hat{k}^{-}_1<\kappa_1<\hat{k}^{+}_1$,
which is shown in Fig.\ref{fig4}(a).

\textbf{Case 2}: $a_1<0$. Two components behave as dark soliton and the amplitude-velocity relations read
\begin{eqnarray}
&&
A_u=\rho_1-\sqrt{\rho^2_1-\frac{2\alpha\rho^2_1(\sqrt{a^2_1+b^2_1}-b_1)}{a^2_1+(b_1-\alpha)^2}  }
=\rho_1-\sqrt{\rho^2_1+ \Delta^{-}_1 },
\\
&&
A_v=\rho_2-\sqrt{\rho^2_2-\frac{2\alpha(1+\rho_1\rho_2)\rho^2_2(\sqrt{a^2_1+b^2_1}-b_1)}{a^2_1+(b_1-\alpha-\alpha\rho_1\rho_2)^2}  }
=\rho_2-\sqrt{\rho^2_2+ \frac{\Delta^{-}_1}{\alpha\kappa_1} },\ \ \ \ \ \ \
\end{eqnarray}
with $\Delta^{-}_1=\frac{\rho^2_1+\alpha\kappa_1\rho^2_2-2\sqrt{\alpha\kappa_1(1+\rho_1\rho_2)\rho^2_1\rho^2_2} }{\rho_1\rho_2}$.
$A_u$ increases in the interval $\hat{k}^{-}_1<\kappa_1<\hat{k}^{u}_0$ with $\hat{k}^{u}_0=\frac{(1+\rho_1\rho_2)\rho^2_1}{\alpha\rho^2_2}$ and decreases in the interval $\hat{k}^{u}_0<\kappa_1<\hat{k}^{+}_1$, while $A_v$ increases in the interval $\hat{k}^{-}_1<\kappa_1<\hat{k}^{v}_0$ with $\hat{k}^{v}_0=\frac{\rho^2_1}{\alpha\rho^2_2(1+\rho_1\rho_2)}$ and decreases in the interval $\hat{k}^{v}_0<\kappa_1<\hat{k}^{+}_1$, which is displayed in Fig.\ref{fig4}(b).
It is noted that the black soliton occurs at two critical points $\hat{k}^{u}_0$ and $\hat{k}^{v}_0$, which correspond to $p_1=\frac{-\alpha\sqrt{(4+3\rho_1\rho_2)\rho_1\rho_2}+{\rm i}(2+\rho_1\rho_2)\alpha}{2(1+\rho_1\rho_2)}$ and $p_1=\frac{\alpha}{2}[-\sqrt{(4+3\rho_1\rho_2)\rho_1\rho_2} + {\rm i}(2+\rho_1\rho_2) ]$ respectively.

\begin{figure}[!htbp]
\centering
{\includegraphics[height=1.6in,width=4.0in]{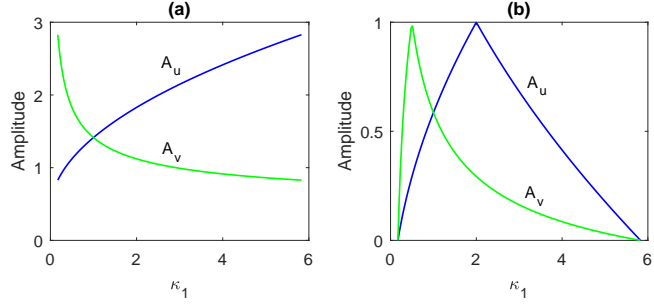}}
\caption{The amplitude-velocity relations with $\alpha=\rho_1=\rho_2=1$: (a) anti-dark soliton ($a_1>0$) and (b) dark soliton ($a_1<0$).
\label{fig4}}
\end{figure}

By taking the limit $a_1\rightarrow 0$, the tau functions have the following expansion
formulae
\begin{eqnarray}
&&f=1-{\rm sgn}(a_1){\rm sgn}(b_1)\left[ 1+2a_1\left( \hat{\eta}_1 - \frac{{\rm i}}{2b_1} \right) \right] + \emph{O}(a^2_1),\ \ \ \ \ \ \ \ \
\\
&&g=1-{\rm sgn}(a_1){\rm sgn}(b_1)\left[ 1+2a_1\left(\hat{\eta}_1 + \frac{{\rm i}}{2b_1 }\frac{\alpha+b_1}{\alpha-b_1} \right) \right] + \emph{O}(a^2_1),\ \ \ \ \ \ \ \ \
\\
&&h=1-{\rm sgn}(a_1){\rm sgn}(b_1)\left[ 1+2a_1\left( \hat{\eta}_1 + \frac{{\rm i}}{2b_1 }\frac{\alpha(1+\sigma\rho_1\rho_2)+b_1}{\alpha(1+\sigma\rho_1\rho_2)-b_1} \right) \right] + \emph{O}(a^2_1),\ \ \ \ \ \ \ \ \
\end{eqnarray}
with $\eta_{10}=x_0-\delta'$ and $\hat{\eta}_1=\frac{\rho_2}{\alpha\rho_1}x-\frac{(1+\sigma\rho_1\rho_2)\alpha\rho_1}{\rho_2b^2_1}t +x_0$.
This suggests that only ${\rm sgn}(a_1){\rm sgn}(b_1)=1$ gives rise to the algebraic soliton solution:
\begin{eqnarray}
&&
u=\rho_1 e^{\mathrm{i}(1+\sigma\rho_1\rho_2) \left(\frac{\rho_2}{\rho_1}x+\frac{\rho_1}{\rho_2}t\right)}
\frac{\hat{\eta}_1 + \frac{{\rm i}}{2b_1 }\frac{\alpha+b_1}{\alpha-b_1} }{\hat{\eta}_1 + \frac{{\rm i}}{2b_1} },
\\
&&
v=\rho_2 e^{\mathrm{i}(1+\sigma\rho_1\rho_2) \left(\frac{\rho_2}{\rho_1}x+\frac{\rho_1}{\rho_2}t\right)}
\frac{\hat{\eta}_1 + \frac{{\rm i}}{2b_1 }\frac{\alpha(1+\sigma\rho_1\rho_2)+b_1}{\alpha(1+\sigma\rho_1\rho_2)-b_1} }{\hat{\eta}_1 - \frac{{\rm i}}{2b_1}},
\end{eqnarray}
with $b_1=\alpha [1+\sigma\rho_1\rho_2+\sqrt{\sigma\rho_1\rho_2(1+\sigma\rho_1\rho_2)}]$.
The usual dark soliton, black soliton and algebraic soliton are illustrated in \ref{fig5} under the different parameters' values.

\begin{figure}[!htbp]
\centering
{\includegraphics[height=1.6in,width=4.0in]{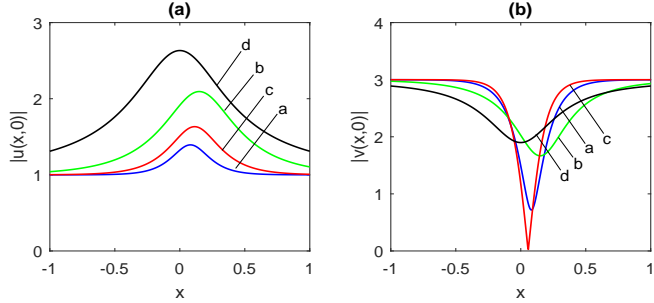}}
\caption{The dark one-soliton  with the parameters $\rho_1=\alpha=-\sigma=\rho_2/3=1$ and $\xi_{10}=x_0=0$: (a) $p_1=\sqrt{2}$, (b) $p_1=\frac{1+(\sqrt{23}-4){\rm i}}{2}$, (c) $p_1=\frac{\sqrt{5}+{\rm i}}{4}$ for $u$ and $p_1=\frac{\sqrt{15}-{\rm i}}{2}$for $v$ ($v$ is a black soliton)  and (d) the algebraic soliton with $p_1= (\sqrt{6}-2){\rm i}$.
\label{fig5}}
\end{figure}

For the dark two-soliton with $N=2$ in (\ref{multi-dark-sol-theorem-01})-(\ref{multi-dark-sol-theorem-03}), we obtain the tau functions
\begin{eqnarray}
&&
f= 1+ d_{11^*} e^{\xi_1+\xi^*_1}+ d_{22^*} e^{\xi_2+\xi^*_2}+d_{11^*}d_{22^*}\Omega_{12} e^{\xi_1+\xi_2+\xi^*_1+\xi^*_2},
\\
&&
g= 1+ d^*_{11^*} K_1 e^{\xi_1+\xi^*_1}+ d^*_{22^*} K_2 e^{\xi_2+\xi^*_2}+d^*_{11^*}d^*_{22^*}K_1K_2 \Omega_{12} e^{\xi_1+\xi_2+\xi^*_1+\xi^*_2},\ \ \ \
\\
&&
h= 1+ d^*_{11^*} H_1 e^{\xi_1+\xi^*_1}+ d^*_{22^*} H_2 e^{\xi_2+\xi^*_2}+d^*_{11^*}d^*_{22^*}H_1H_2 \Omega_{12} e^{\xi_1+\xi_2+\xi^*_1+\xi^*_2},\ \ \ \
\end{eqnarray}
with
\begin{equation*}
d_{ii^*}=\frac{{\rm i}p_i}{p_i+p^*_i},\ \
K_i=-\frac{p_i-{\rm i}\alpha}{p^*_i+{\rm i}\alpha},\ \
H_i=-\frac{p_i-{\rm i}\alpha(1+\sigma\rho_1\rho_2)}{p^*_i+{\rm i}\alpha(1+\sigma\rho_1\rho_2)},\ \
\Omega_{12}=\frac{|p_1-p_2|^2}{|p_1+p^*_2|^2}.
\end{equation*}

Since the dark one-soliton solution exhibits dark and anti-dark soliton, the dark two-soliton solution may possess
three types of the interaction: dark-dark solitons, dark-anti-dark solitons,
and anti-dark-anti-dark solitons, which are displayed in Fig.\ref{fig6}.
From the constraint conditions (\ref{constraint-condition}), if we assume the same velocity $\kappa_i$,
one cannot get different value of $p_i$.
Hence there does not exist bound state for the dark-soliton in the MT model.

\begin{figure}[!htbp]
\centering
{\includegraphics[height=2.6in,width=5.0in]{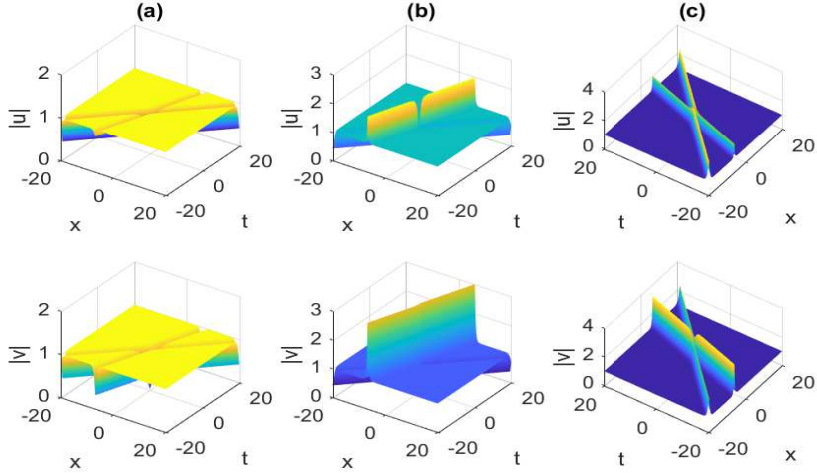}}
\caption{The dark two-soliton  with the parameters $\rho_1=\rho_2=\alpha=\sigma=1$ and $\xi_{10}=\xi_{20}=0$: (a) dark-dark soliton $p_1=1+{\rm i}$, $p_2=\sqrt{2}+2{\rm i}$; (b) dark-anti-dark soliton $p_1=1+{\rm i}$, $p_2=-\sqrt{2}+2{\rm i}$; (c) anti-dark-anti-dark soliton $p_1=-1+{\rm i}$, $p_2=-\sqrt{2}+2{\rm i}$.
\label{fig6}}
\end{figure}

 \section{Concluding Remarks}
In this paper, we give bilinear formulation of the massive Thirring model which is  missing in the literature. Based on the bilinear equations, we constructed both the bright and dark soliton solutions to the massive Thirring model under VBC and NVBC, respectively, via the KP hierarchy reduction technique. Especially, we have shown that the discrete KP equation can generate a set of four bilinear equations, which are reduced to bilinear equations corresponding to the dark soliton of the MT model. We have also listed one- and two- bright and dark soliton solutions and analyzed their propertied in details.

We should point out that, even if we give the bright and dark soliton solution to the MT model in the light cone coordinates, one can easily obtain the corresponding soliton solutions  in the laboratory coordinates. Moreover, based on the bilinear formulation established in the present work,  we can also construct general breather and rogue wave solutions of the MT model (\ref{MTa})--(\ref{MTb})  via the KP hierarchy reduction method. We will report our results elsewhere in the near future. 

Finally, we comment that, since we have established a link between the discrete KP equation and the MT model in terms of tau functions and bilinear equations, similar to our recent work on modified Camassa-Holm (mCH) equation \cite{ZYFmCH-discrete}, it paves a way for constructing integrable semi-discrete and fully discrete MT model, which definitely deserves immediate investigation.

\section*{Acknowledgement}
We thank Prof. Dmitry Pelinovsky for drawing us attention of the massive Thirring model and useful comments on our first draft.
JC's work was supported from the National Natural
Science Foundation of China (NSFC) (No. 11705077).
BF's work is
partially supported by National Science Foundation (NSF) under Grant No. DMS-1715991 and U.S. Department of Defense (DoD), Air Force for Scientific
Research (AFOSR) under grant No. W911NF2010276.


\begin{thebibliography}{99}

\bibitem{Thirring}  W.E. Thirring, A soluble relativistic field theory, Ann. Phys. 3, 91--112, (1958).

\bibitem{Mikhailov} A.V Mikhailov, Integrability of the two-dimensional Thirring model, JETP Lett. 23,  320--23 (1976),.

\bibitem{Orfanidis}  S.J. Orfanidis, Soliton solutions of the massive Thirring model and the inverse scattering transform,
Phys. Rev. D 14, 472 (1976).




\bibitem{Kuznetsov} E.A. Kuznetzov and A.V. Mikhailov, On the complete integrability of the two-dimensional classical Thirring
model, Theor. Math. Phys. 30  193--200 (1977).

\bibitem{Kawata79}  K. Kawata, T. Morishima, and H. Inoue,  Inverse scattering method for the two-dimensional massive Thirring model, J. Phys. Soc. Japan 47, 1327--34, (1979).


\bibitem{WadatiMT83}  M. Wadati, General solution and Lax pair for 1D classical massless Thirring model, J. Phys. Soc. Japan 52, 1084--85 (1983). 

\bibitem{KaupLakoba}  D. J.  Kaup and T. Lakoba,
The squared eigenfunctions of the massive Thirring model in laboratory coordinates, J. Math. Phys. 37 308-23 (1996).

\bibitem{Villarroel:1991}
J.~Villarroel, The {DBAR} problem and the {T}hirring model,  {Stud. Appl. Math.} 84, 207--220(1991).

\bibitem{DmkitryMTIST}  D. Pelinovsky and A. Saalmann, Inverse Scattering for the Massive Thirring Model, 
In: Miller P., Perry P., Saut JC., Sulem C. (eds) Nonlinear Dispersive Partial Differential Equations and Inverse Scattering. Fields Institute Communications, vol 83. Springer, New York, NY.
497--528 (2019). 


\bibitem{KaupNewell}  D.J. Kaup and A.C. Newell, On the Coleman correspondence and the solution of the Massive Thirring model,
{Lett. Nuovo Cimento} 20, 325--331 (1977).

\bibitem{Lee:1993}
J.-H. Lee.
$n\times n$ {Z}akharov--{S}habat system of the form
  $(d\psi/dx)(z^2-1/z^2){J}\psi+(z{Q}+{P}+{R}/z)\psi$.
 In A.~S. {F}okas, D.~J. {K}aup, A.~C. {N}ewell, and V.~E. {Z}akharov,
  editors, {Nonlinear {P}rocesses in {P}hysics}, 118--121,
Springer, Berlin, 1993.

\bibitem{Lee:1994}
J.-H. Lee.
\newblock Solvability of the derivative nonlinear {S}chr\"odinger equation and
  the massive {T}hirring model,
\newblock {\em Theoret. and Math. Phys.}, 99, 617--621, (1994).

\bibitem{Prikarpatskii:1979}
A.~K. Prikarpatskii and P.~I. Golod,
\newblock Periodic problem for the classical two-dimensional {T}hirring model.
\newblock {\em Ukrainian Math. J.}, 31, 362--367 (1979).

\bibitem{Prikarpatskii:1981}
A.~K. Prikarpatskii.
\newblock Geometrical structure and {B}\"acklund transformations of nonlinear
  evolution equations possessing a {L}ax representation.
\newblock {\em Theoret. and Math. Phys.}, 46, 249--256 (1981).

\bibitem{Franca13}  G. S. Franca, J. F. Gomes, A. H. Zimerman,
The algebraic structure behind the derivative nonlinear Schr\"odinger equation,
J. Phys. A: Math. Theor. 46,  305201 (2013).

\bibitem{DegasperisMT}  A. Degasperis, Darboux polynomial matrices:the classical massive Thirring model as a study case,
J. Phys. A: Math. Theor. 48, 235204 (2015).

\bibitem{Date}  E. Date, On a construction of multi-soliton solutions of the Pohlmeyer-Lund-Regge system and the classical massive Thirring model
Proc. Japan Acad. Ser. A Math. Sci. 55,  278-281 (1979).

\bibitem{Shnider84}  D. David, J. Harnad, and S. Shnider, Multisoliton solutions to the Thirring model through the
reduction method, Lett. Math. Phys. 8, 27--37 (1984).

\bibitem{Alonso:1984}
L.~Martinez Alonso,
\newblock Soliton classical dynamics in the sine-Gordon equation in terms of
  the massive Thirring model.
 {\em Phys. Rev. D} 30, 2595--2601 (1984).


\bibitem{BarashenkovGetmanov:1987}
I.~V. Barashenkov and B.~S. Getmanov,
 Multisoliton solutions in the scheme for unified description of
  integrable relativistic massive fields. Non-degenerate $sl(2, C)$ case.
 {\em Comm. Math. Phys.} 112, 423--46 (1987).

\bibitem{BarashenkovGetmanovKovtun:1993}
I.~V. Barashenkov, B.~S. Getmanov, and V.~E. Kovtun,
\newblock The unified approach to integrable relativistic equations: Soliton
  solutions over nonvanishing backgrounds. {I}.
 {\em J. Math. Phys.}  34, 3039--3053 (1993).

\bibitem{BarashenkovGetmanov:1993}
I.~V. Barashenkov and B.~S. Getmanov,
\newblock The unified approach to integrable relativistic equations: Soliton
  solutions over nonvanishing backgrounds. {II}.
\newblock {\em J. Math. Phys.} 34, 3054--72 (1993).

\bibitem{Talalov:1987}
S.~V. Talalov,
\newblock Hamiltonian structure of  {T}hirring {L}iouville model. Singular solutions.
\newblock {Theor. Math. Phys.} 71, 588--97 (1987).

\bibitem{Vaklev:1996}
Y.~Vaklev,
\newblock Soliton solutions and gauge equivalence for the problem of
  {Z}akharov--{S}habat and its generalizations.
\newblock {J. Math. Phys.} 37, 1393--1413 (1996).


\bibitem{DateMT2}  E. Date. On quasi-periodic solutions of the field equation of the classical massive Thirring
model,  Progr. Theoret. Phys. 59, 265--73 (1978).

\bibitem{MTquasi2}  A. K. Prikarpatskii and P. I. Golod, Periodic problem for the classical two-dimensional
Thirring model. Ukrainian Math. J. 31, 362--367 (1979).

\bibitem{Bikbaev}  R. F. Bikbaev, Finite-gap solutions of the massive Thirring model, Teoret. Mat. Fiz. 63, 377--87
(1985).

\bibitem{Wisse85}  M.A. Wisse, Darboux Coordinates and Isospectral Hamiltonian
Flows for the Massive Thirring Model,  Lett. Math. Phys. 28 287-94 (1993).


\bibitem{MTEnolskii}  V.Z. Enolskii, F. Gesztesy, H. Holden, The classical massive Thirring system revisited, 200
In Stochastic processes, physics and geometry: new interplays. I. A volume in honor of Sergio
Albeverio (ed. F. Gesztesy, H. Holden, J. Jost, S. Paycha, M. Rockner and  S. Scarlatti).
Canadian Mathematical Society Conference

\bibitem{MTElbeck}  J. C. Elbeck, F. Enolskii, H. Holden,
The hyperelliptic $xi$-function and the integrable massive Thirring model, Proc. R. Soc. Lond. A 459 1581-1610 (2003).



\bibitem{DWA-MT2015}  A. Degasperis, S. Wabnitz, A.B. Aceves, Bragg grating rogue wave, Phys. Lett. A 379, 1067?1070 (2015) .


\bibitem{HeMT2015}  L. Guo, L. Wang, Y. Cheng, and J. He, High-order rogue wave solutions of the classical massive
Thirring model equations, Commun. Nonlinear Sci. Numer. Simulat. 52, 11 (2017).

\bibitem{ShihuaMT}  Y. Ye, L. Bu, C.  Pan, Shihua Chen, D. Mihalache, F. Baronio,
Super rogue wave states in the classical massive Thirring model system,
Rom. Rep. Phys. 73, 117 (2021).

\bibitem{Sipe}  C.M. de Sterke and J.E. Sipe, Gap solitons, Prog. Opt. 33, 203--260 (1994).

\bibitem{BarashenkovGetmanov:2019}
N. V. Alexeeva, I. V. Barashenkov and A. Saxena A,
\newblock Spinor solitons and their PT-symmetric
offspring
\newblock {\em Ann. Phys.} 403, 198-223 (2019).

\bibitem{NijhoffMT1}
F.~W. Nijhoff, H.~W. Capel, G.~R.~W. Quispel, and J.~van~der Linden,
\newblock The derivative nonlinear {S}chr\"odinger equation and the massive
  {T}hirring model.
\newblock {\em Phys. Lett. A} 93, 455--458 (1983).


\bibitem{NijhoffMT2}  F.W. Nijhoff, H. W. Capel and G. R. W. Quispel, Integrable lattice version of the massive Thirring model and its linearization, Phys. Lett. A 98, 83-86 (1983).

\bibitem{DmkitryMTdiscrete}  J. Nalini and D. E. Pelinovsky,
Integrable semi-discretization of the massive Thirring system in laboratory coordinates,
J. Phys. A   52, 03LT01 (2019).


\bibitem{DmitryXu2019}  T. Xu and D.E. Pelinovsky, Darboux transformation and soliton solutions of the semi-discrete
massive Thirring model, Phys. Lett. A 383, 125948 (2019).

\bibitem{Hirotabook} R. Hirota, 2004 \textit{The direct method in soliton
theory}, (Cambridge Univ. Press, Cambrige, 2004).

\bibitem{JM} M. Jimbo and T. Miwa, Solitons and infinite-dimensional Lie
algebras, {Publ. Res. Inst. Math. Sci.}  19,  943--1001 (1983).


\bibitem{FengvNLS}  B.-F. Feng, General N-soliton solution to a vector
nonlinear Schr\"odinger equation, J. Phys. A 47, 355203 (2014).

\bibitem{FengShen_ComplexSPE} S. Shen, B.-F. Feng and Y. Ohta, From the real
and complex coupled dispersionless equations to the real and complex short
pulse equations, {Stud. Appl. Math.} 136, 64--88 (2016).

\bibitem{FMO_ComplexSPE} B.-F. Feng, K. Maruno,Y. Ohta, Geometric Formulation and Multi-dark Soliton Solution to the Defocusing Complex Short Pulse Equation
, \emph{Stud. Appl. Math.} 138, 343--367 (2017).


	
\bibitem{MiyakeOhtaGram} S. Miyake, Y. Ohta and J. Satsuma, A
representation of solutions for the KP hierarchy and its algebraic
structure, {J. Phys. Soc. Jpn.} {59}, 48--55 (1990).

\bibitem{Hirota-1981} R.Hirota, Discrete analogue of a generalized {T}oda equation, 
{J.~Phys. Soc. Japan} {50}, 3785--3791  (1981).

  \bibitem{Miwa-1982}
T. Miwa, On {H}irota's difference equations,
{Proc. Japan Acad. Ser.~A Math. Sci.} {58},  9--12 (1982).

\bibitem{OHTI93} Y. Ohta, R. Hirota, S.  Tsujimoto
and T. Imai, Casorati and discrete Gram type determinant representations
of solutions to the discrete KP hierarchy, {J. Phys. Soc. Jpn.} {62}, 1872--1886 (1993).


	\bibitem{MatsunoFL2}
	Y. Matsuno, A direct method of solution for the Fokas-Lenells derivative nonlinear Schr\"odinger equation: II. Dark soliton solutions, { Phys A}  45,475202(2012).
		
	\bibitem{Matsuno-mCH}
	Y. Matsuno, Smooth and singular multisoliton solutions of a modified Camassa-Holm equation with cubic nonlinearity and linear dispersion, { Phys A} 47,125203 (2014).
	
	\bibitem{ZYFmCH-discrete}
	H.-H. Sheng, G.-F. Yu and B.-F. Feng, An integrable semi-discretization of the modified Camassa-Holm equation with linear dispersion term, arXiv:2110.15876 [nlin.SI].

\end{thebibliography}
\end{document}